\documentclass[aps,pra,twocolumn,showpacs,superscriptaddress,groupedaddress]{revtex4-2}
\usepackage[utf8]{inputenc}
\usepackage{amsmath, amsthm, amssymb, amsfonts,enumitem}
\usepackage[table]{xcolor}
\usepackage{colortbl}
\usepackage{adjustbox}      
\usepackage{ytableau}

\definecolor{lightgreen}{rgb}{0.8,1,0.8}
\definecolor{lightred}{rgb}{1,0.8,0.8}
\definecolor{lightorange}{rgb}{1,0.9,0.7}

\usepackage[makeroom]{cancel}
\usepackage{graphicx}
\usepackage{bm}
\usepackage{dsfont}
\usepackage{mathtools}
\usepackage{booktabs}
\usepackage{thm-restate}
\usepackage[colorlinks=true, linkcolor=blue, urlcolor=blue,citecolor=blue]{hyperref}
\usepackage{orcidlink}
\usepackage{mathtools}
\usepackage{siunitx}

\usetikzlibrary{matrix, decorations.pathreplacing, positioning,calc}

\DeclarePairedDelimiter{\abs}{\lvert}{\rvert}

\DeclarePairedDelimiter{\of}{\lparen}{\rparen}
\DeclarePairedDelimiter{\sof}{\lbrack}{\rbrack}

\providecommand\given{}

\newcommand\SetSymbol[1][]{%
    \nonscript\:#1\vert
    \allowbreak
    \nonscript\:
    \mathopen{}}
\DeclarePairedDelimiterX\Set[1]\{\}{%
    \renewcommand\given{\SetSymbol[\delimsize]}
    #1
}

\newcommand{\U}[1]{\mathrm{U}_{#1}}

\DeclareMathOperator{\imm}{Im}
\DeclareMathOperator{\per}{per}

\DeclareMathOperator{\tr}{tr}
\DeclareMathOperator{\SEP}{SEP}

\DeclareMathOperator{\id}{id}

\newcommand{\bra}[1]{\mathinner{\langle #1|}}
\newcommand{\ket}[1]{\mathinner{|#1\rangle}}

\newcommand{\ketbra}[2]{\mathinner{| #1 \rangle\!\langle #2 |}}
\newcommand{\dyad}[1]{| #1\rangle \langle #1|}

\renewcommand{\k}{\kappa}
\newcommand{\kSEP}{\kappa\operatorname{-SEP}}
\newcommand{\one}[0]{\mathds{1}}

\newcommand{\cdn}{(\C^d)^{\otimes n}}

\newcommand{\R}{\mathds{R}}
\newcommand{\C}{\mathds{C}}

\DeclareUnicodeCharacter{202C}{\^{i}}

\makeatletter
\newcommand{\vast}{\bBigg@{4}}
\newcommand{\Vast}{\bBigg@{5}}
\makeatother
\newtheorem{theorem}    {Theorem}

\newtheorem{cor}    {Corollary}
\newtheorem{proposition}[theorem]{Proposition}
\newtheorem{observation}[theorem]{Observation}

\newtheorem{definition}    [theorem]{Definintion}

\usepackage{soul,xcolor}

\begin{document}

\title{
Detection of many-body entanglement partitions in a quantum computer
}
\author{Albert Rico${}^{1\orcidlink{0000-0001-8211-499X}}$, Dmitry Grinko${}^{2,3,4\orcidlink{0000-0001-6438-1998}}$, Robin Krebs$^{5\orcidlink{0009-0006-0397-9578}}$, Lin Htoo Zaw${}^{6\orcidlink{0000-0003-1559-6903}}$}
\affiliation{$^1$GIQ - Quantum Information Group, Department of Physics, Autonomous University of Barcelona, Spain}%
\affiliation{$^2$QuSoft, Amsterdam, The Netherlands}
\affiliation{$^3$Institute for Logic, Language and Computation, University of Amsterdam, The Netherlands} 
\affiliation{$^4$Korteweg-de Vries Institute for Mathematics, University of
Amsterdam, The Netherlands}
\affiliation{$^5$Quantum Computing, Technische Universität Darmstadt}
\affiliation{$^6$Centre for Quantum Technologies, National University of Singapore, 3 Science Drive 2, Singapore 117543}
\date{\today}
\begin{abstract}
We present a method to detect entanglement partitions of multipartite quantum systems, by exploiting their inherent symmetries. Structures like genuinely multipartite entanglement, $m$-separability and entanglement depth are detected as very special cases. 
This formulation enables us to characterize all the entanglement partitions of all three- and four- partite states and witnesses with unitary and permutation symmetry. In particular, we find and parametrize a complete set of bound entangled states therein. 
For larger systems, we provide a large family of analytical witnesses detecting many-body states of arbitrary size where none of the parties is separable from the rest. This method relies on weak Schur sampling with projective measurements, and thus can be implemented in a quantum computer.  
Beyond Physics, our results extend to the mathematical literature: we establish new inequalities between matrix immanants, and characterize the set of such inequalities for matrices of size three and four.

\end{abstract}

\maketitle

\section{Introduction}
Over the past decades, a large improvement has been done in the control over quantum systems of increasing size~\cite{advancesHDent_Erhard2020,EfficientLargeScaleMBdyn_Artaco2024}. This advances the development of quantum devices, for which a key feature is the presence of quantum entanglement~\cite{EntDetRev_Guhne2009}. Yet, multiple systems can be entangled in different structures, and each is beneficial for different applications. The most resourceful case is arguably genuinely multipartite entanglement, where no subsystem is separable to the rest, and a variety of methods have been derived to detect it~\cite{Dur_SepGME2000,Guhne_SepGME2010,MHuber_2010HDdetectGME}. However, detecting finer structures is a notorius challenge~\cite{Ananth_CritkSep2015,Guhne2006EnergyKsep,Szalay2019KsepMprod,Gao_2013EfficientKsep,Seevinck2001Suf3EntKsep,Ghhne2005MultiESpinChKsep}.

At the same time, rapid progress in quantum control is increasingly approaching the physical realization of quantum computers~\cite{haffner2008quantumCompIon,Preskill2018quantumcomputingin,devoret2013superconductingComp}. Beyond enabling efficient algorithms for classically difficult problems~\cite{grover1996fast,ShorFactorization_1997,arute2019quantumCompGoogle,zhong2020quantumComp}, quantum computers are also powerful tools for analysing quantum systems, for instance in many-body simulation~\cite{blatt2012quantumSimulIons,fauseweh2024quantumSimulation} and quantum chemistry~\cite{aspuru2005simulatedChemistry,colless2018computationChemistry}. In particular, quantum hardware is a natural platform where to certify the correlations of quantum systems. While well-established techniques such as randomized measurements are effective for bipartite systems~\cite{Neven_SymmetryResMomentsPT2021,Elben2020,Elben22Toolbox,Cieslinski_2024AnalyzingQSRandomMeas}, less is known for more involved entanglement structures.

Here we develop a technique to detect the entanglement structure of a multipartite system, which is tractable to implement in a quantum computer. By exploiting entanglement symmetries under permutations and unitaries (Observation~\ref{obs:KSepAlpha}), we detect the entanglement partition a of quantum system (Definition~\ref{def:K-sep}). We characterize the entanglement partitions of three- and four-partite invariant systems (Proposition~\ref{prop:CharactWitN=3} and Table~\ref{tab:Main4Partite}), and identify the families of bound entangled states therein. Exploiting symmetries further, we derive analytical criteria to certify that no local parties are separable from the rest in many-body systems (Theorem~\ref{thm:exact_(1,n-1)_sep}). Different methods for practical estimation of the criteria introduced are compared in Table~\ref{tab:Scaling}. Besides contributing to Physics, the results presented lead to new symmetric matrix inequalities, a long-standing problem in the mathematics community (Proposition~\ref{prop:WitnessImmanant}). In particular we characterize the set of immanant inequalities for positive semidefinite matrices of size three and four.

\begin{figure}[tbp]
    \centering
    \includegraphics[width=0.9\linewidth]{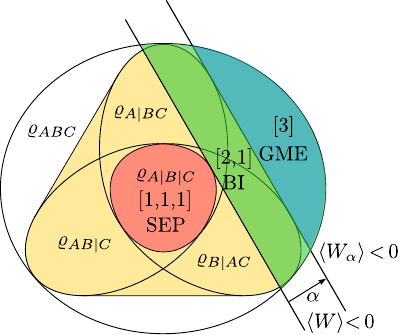}
    \caption{{\bf Structures of three-partite separability.} Three-partite states can be fully separable (SEP), with separability partition $\k=[1,1,1]$; biseparable ($\varrho_{AB|C}$, $\varrho_{A|BC}$, $\varrho_{AC|B}$ and their convex combination), with partition $\k=[2,1]$; or genuinely multipartite entangled (GME), with $\k=[3]$. We shift symmetric witnesses $W$ detecting biseparable states (green area) by a parameter $\alpha_\k$, to construct new witnesses detecting states outside of certain separability partitions (blue area).}
    \label{fig:GME}
\end{figure}

\section{Main goal}
Consider a three-partite quantum system. A state is fully separable if it can be written as 
$\varrho_{ABC}=\sum_i p_i\varrho_A^i\otimes\varrho_B^i\otimes\varrho_C^i$. It is biseparable if it can be written as 
\begin{equation}
    \varrho_{ABC}=\sum_{i,j,k}p_i\varrho_A\otimes\varrho_{BC}+p_j\varrho_B\otimes\varrho_{AC}+p_k\varrho_{AB}\otimes\varrho_{C},
\end{equation}
where $p_i,p_j,p_k\geq 0$ and $\sum_{i,j,k}p_i+p_j+p_k=1$. And it is genuinely multipartite entangled if it is neither fully separable nor biseparable (see Fig.~\ref{fig:GME}). The case of larger systems is more involved, as separability structures could be distributed in multiple ways. Here we consider the following entanglement structures, which fine-grain the entanglement properties of the state in hand~\cite{Huber2009TwoCompEntParts,Huber2013EntVecFormPartEnt,Ren2021MetrologicalDetPartitions,Garcia2023ExpEntParts,Lu2018EntPart,zhou2019MPEminRes}:
\begin{definition}\label{def:K-sep}
Given a partition $\k = [k_1|...|k_{m}]$ with $k_1\geq ...\geq k_m$ and $k_1+...+k_m=n$, an $n$-partite state $\varrho$ is $\k$-separable, written $\varrho \in \kSEP$, if
\begin{equation}\label{eq:VecKsep}
    \varrho = \sum_ip_i\varrho^i_{K_1^i}\otimes...\otimes\varrho^i_{K_m^i},
\end{equation}
where each factor $\varrho^i_{K_r}$ is shared in a disjoint collection $K_r^i \subseteq \{1,\dotsc,n\}$ of $\abs{K_r^i} = k_r$ local parties.

\end{definition}
For example, the state $\varrho_A\otimes\varrho_{BCD}$ is $[3|1]$-separable and the state $\varrho'=(\varrho_{AB}\otimes\varrho_{CD}+\varrho_{AC}\otimes\varrho_{BD})/2$  is $[2|2]$-separable. 

Determining the separability partition $\k$ of a state is crucial for multipartite quantum protocols such as distributed computing and network communiaction. In particular, two paradigmatic notions of multipartite entanglement stem from it: {\em $k_1$-producibility} or {\em entanglement depth}, namely the largest number $k_1$ of systems among the terms decomposing $\varrho$~\cite{Sorensen2001EntDepth}; and {\em $m$-separability, namely the number $m$ of composite systems in a product state}~\cite{Gao_2013EfficientKsep}. These forms of entanglement are dual to each other~\cite{Szalay2019KsepMprod}, and have been considered in several forms~\cite{NetworkGMPE_Navascues2020} for practical applications. While a variety of methods have been long studied to detect both non-$m$-separability~\cite{Dur_SepGME2000,Seevinck2001Suf3EntKsep,Ghhne2005MultiESpinChKsep,Guhne2006EnergyKsep,MHuber_2010HDdetectGME,Gao_2013EfficientKsep} and non-$k_1$-producibility~\cite{Vitagliano2011SpinSqueezKprod,Vitagliano2018PlanarEntDepth,Aloy2019DVIedepth}, much less is known about the more general characterization of Definition~\ref{def:K-sep}~\cite{Huber2009TwoCompEntParts,Huber2013EntVecFormPartEnt,Ren2021MetrologicalDetPartitions,Garcia2023ExpEntParts}.

\section{Methodology}
\begin{figure*}[tbp]
    \centering
    \includegraphics[width=\linewidth]{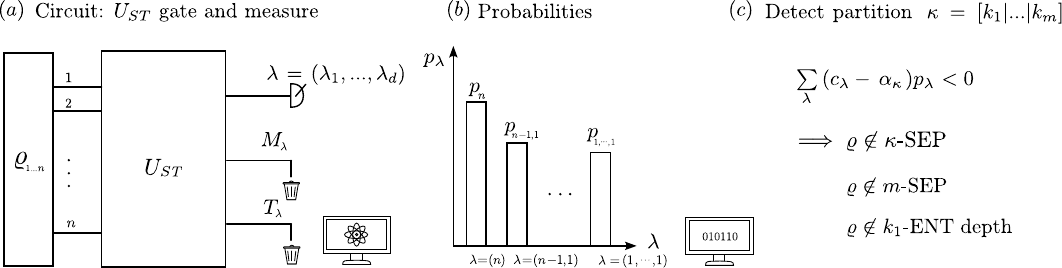}
    \caption{{\bf Detection of entanglement partitions $\kappa$ in $n$-partite systems.} $(a)$ We perform weak Schur sampling on the $n$-partite state describing the system, for instance implementing the Schur transform gate ($U_{ST}$) and measuring the register $\lambda$ while discarding the registers $M_\lambda$ and $T_\lambda$. This step is tractable in a quantum computer, and efficiency can be gained with more direct methods of measuring $\lambda$ such as generalized phase estimation (see Table~\ref{tab:Scaling}). 
    $(b)$ We repeat the procedure multiple times, and store in a classical computer the probabilities $p_\lambda$ of obtaining each outcome $\lambda$. These probabilities are the projections of $\varrho$ onto each irreducible subspace labelled by $\lambda$. $(c)$ Based on the probabilities $p_\lambda$, we detect states that are not in certain separability partitions $\kappa$. In particular, the first component $k_1$ of $\kappa$ determines the entanglement depth of $\varrho$, and the number of components $m$ in $\kappa$ determines its separability length.}
    \label{fig:Circuit}
\end{figure*}
To detect $\k$-separability efficiently, we will exploit two symmetries of the problem:

- {\em Permutation symmetry:} it follows from Eq.~\eqref{eq:VecKsep} that the separability partition $\k$ of a quantum state $\varrho$ cannot be affected by the action $V_\pi$ of a permutation $\pi\in S_n$,
\begin{equation}
    V_\pi\varrho V_\pi^\dag \in \kSEP \iff \varrho \in \kSEP\,.
\end{equation}

- {\em Unitary symmetry:} It is a defining property of entanglement that it cannot be affected by the action of local unitary gates. Here we will use entanglement symmetry under the diagonal action of the unitary group,
\begin{equation}
    U^{\otimes n}\varrho {U^{\otimes n}}^\dag  \in \kSEP \iff \varrho \in \kSEP\,.
\end{equation}

Extensive work has been done in the study of entanglement exploiding separately unitary~\cite{Eggeling2001UnSymmEnt} and permutation~\cite{Toth2009EntPerSym} symmetries. Here we use the combination between both symmetries above, which are exactly those of the Hilbert space decomposition through Schur--Weyl duality,
\begin{equation}
\label{eq:SW_duality}
    \cdn\simeq \bigoplus_{\lambda \vdash_d n} \mathcal U_\lambda\otimes \mathcal S_\lambda\,,
\end{equation}
where $\lambda=(\lambda_1,\dotsc,\lambda_{d})$ is a partition of $n$ in the sense that $\lambda_1+\dotsc+\lambda_{d}=n$ and $\lambda_1 \geq \dotsc \geq \lambda_{d} \geq 0$. The notation $\lambda \vdash_d n$ means that $\ell(\lambda)\leq d$, where $\lambda$ is a partition of $n$ elements and $\ell(\lambda)$ is the number of non-zero components. Here $\mathcal U_\lambda$ is an irreducible representation (irrep) of the unitary group $U(d)$, and $\mathcal S_\lambda$ is an irrep of the symmetric group $S_n$. The projectors $\Pi_\lambda$ onto isotypic components $\mathcal U_\lambda\otimes \mathcal S_\lambda$ are called \emph{Young projectors}. The dimension $d_\lambda$ of the irrep $\mathcal S_\lambda$ is given by evaluating the symmetric group character at the identity, $d_\lambda = \chi_\lambda(\id)$, and it is given by \emph{hook-length} formula. 
The basis change unitary matrix of Equation~(\ref{eq:SW_duality}) is called \emph{Schur transform}, and it can be implemented efficiently on a quantum computer~\cite{EfficientQCirQST_Bacon2006,Krovi2019efficienthigh,grinko2025mixed,burchardt2025hdschur}. Measurement of Young projectors is called \emph{weak Schur sampling} and it can be done efficiently via the Schur transform or via Generalized Phase estimation~\cite{harrow2005thesis, cervero2023weakschursamplinglogarithmic}.

We consider the detection of $\k$-separable states using the linear combination $\sum_{\lambda \vdash_d n} c_\lambda p_\lambda$, where $c_\lambda\in\R$, of probabilities $p_\lambda = \tr(\varrho \Pi_\lambda)$ of obtaining different partitions in measurement outcomes.
Namely, we shall look at witnesses of the form $W=\sum_{\lambda \vdash_d n} c_\lambda  \Pi_\lambda$.
This ansatz is motivated by a wide supply of full-separability witnesses involving projections onto irreducible subspaces, of the form $c_\lambda\Pi_\lambda-c_\mu\Pi_\mu$~\cite{MaassenSlides,long}, available from long-studied so-called {\em matrix immanant inequalities}~\cite{Bhatia_MatAn1997}.
For a given $W$, by considering the minimum expectation value $\alpha_{\k} = \min_{\varrho' \in \kSEP} \tr(\varrho'W)$ over $\kSEP$ and the fact that $\Pi_\lambda$ resolve to the identity, $\k$-separability witnesses are given by
\begin{equation}\label{eq:WitsKSEP}
\sum_{\lambda \vdash_d n} ( c_\lambda - \alpha_{\k} ) p_\lambda < 0 \implies \varrho\not\in \kSEP\,.
\end{equation}
To find the optimal value $\alpha_{\k}$, we will use the permutation invariance of the projectors $\Pi_\lambda$:
\begin{observation}\label{obs:KSepAlpha}
The value $\alpha_{\k}$ is given by
\begin{equation}
    \alpha_{\k} = \min_{\ket{\psi}=\ket{\psi}_{k_1}\otimes...\otimes\ket{\psi}_{k_m}}\bra{\psi}W\ket{\psi},
\end{equation}
where $\ket{\psi}_{k_r}$ is shared among parties $k_{r-1}+1,...,k_r$. 
\end{observation}
\begin{proof}
    The value $\alpha_\k$ is by assumption the smallest expectation value of $W$ over $\k$-separable states.
    By convexity, we need only consider pure states $\ket{\psi} = \ket{\psi}_{k_{n_1}}\otimes...\otimes\ket{\psi}_{k_{n_m}} \in \kSEP$.
    Since the projectors $\Pi_\lambda$ are permutation-invariant, the expectation value $\bra{\psi}W\ket{\psi}$ does not change if all terms decomposing the state $\ket{\psi}$ are permuted to be of the ordered form $\ket{\psi}_{k_1}\otimes...\otimes\ket{\psi}_{k_m}$, which completes the proof.
\end{proof}
Let us consider quantum states $\varrho_{\text{sym}}$ that are invariant under permutations and copies of unitaries, 
\begin{equation}\label{eq:PermuUnitSymm}
    V_\pi \varrho_{\text{sym}} V_\pi^\dag = U^{\otimes n}\varrho_{\text{sym}} {U^\dag}^{\otimes n} = \varrho_{\text{sym}}
\end{equation}
for all $\pi\in S_n$ and $U\in\mathrm{U}(d)$. These can be written as a convex combination of the projectors $\Pi_\lambda$ and are not $\k$-separable if and only if they are detected by some witnesses of the form~\eqref{eq:WitsKSEP}. This is because the minimization of $\tr(W\varrho)$ over $\k$-separability witnesses can be restricted without loss of generality to fully symmetric witnesses, which can be written as~\eqref{eq:WitsKSEP}. For each system size $n$ and $d$, the number of such witnesses that are independent is finite, since the state space is given by a polytope.

\section{Entanglement characterization}
By exploiting the permutation symmetry and unitary symmetry of projectors $\Pi_\lambda$, we characterize the set of three-partite witnesses with these symmetries.
\begin{proposition}\label{prop:CharactWitN=3}
For three-partite systems, all nontrivial (i.e.,~excluding $p_\lambda \geq 0$) criteria of the form~\eqref{eq:WitsKSEP} for separability partitions are 
\begin{align}
 \varrho \in \SEP &\implies
\begin{cases}
    p_{2,1} - 4p_{1,1,1} \geq 0\,,\\
    3p_{3} - p_{2,1} + p_{1,1,1} \geq 0\,,
\end{cases} \\
\varrho \in [2,1]\text{-}\SEP &\implies\quad p_{2,1} - 2p_{1,1,1} \geq 0\,.
\end{align}
and convex combinations thereof.
These criteria are necessary and sufficient for $\k$-separability of tripartite states with symmetry~\eqref{eq:PermuUnitSymm}.
\end{proposition}
The proof is given in Appendix~\ref{app:Charact3partite}.
The first witness and $4p_{3} - p_{2,1} \geq 0$, a convex combination of the first two witnesses, are known and arise from immanant inequalities~\cite{MaassenSlides}.
The third one is to our knowledge new, although in principle it may be derived independently from the results in~\cite{Eggeling2001UnSymmEnt}.
In fact, it detects three-partite states with positive partial transpositions (PPT), while the other ones cannot. We thereby identify the distinguished one-parameter family of states
\begin{equation}\label{eq:3partPPT}
    \varrho_p = p\frac{\Pi_{(3)}}{\tr \Pi_{(3)} } + (1-p)\frac{\Pi_{(2,1)}}{\tr \Pi_{(2,1)} }.
\end{equation}
Using Proposition~\ref{prop:CharactWitN=3}, we verify that all such states are $[2,1]$-separable for $0\leq p\leq 1$, bound entangled for $1/4<p\leq 1/5$, and fully separable for $1/5<p\leq 1$.

\begin{table}[h!]
    \centering
\begin{tabular}{c | c c c c c | c | c c c c}
    \hline
    $W$ & $p_4$ & $p_{3,1}$ & $p_{2,2}$ & $p_{2,1^2}$ & $p_{1^4}$ &
    $d$
    & $\alpha_{1|3}$
    & $\alpha_{2|2}$
    & $\alpha_{1|1|2}$ & $\alpha_{1^4}^{PPT}$ 
    \\
    \hline
    $W_1$ & $9$ & $-1$ & 0 & 0 & 0
    & 2
    & \cellcolor{lightred}$-1$
    & \cellcolor{lightorange}$-1$
    & \cellcolor{lightorange}$-1$ & $0$  \\

    $W_2$ & $0$ & $9$ & $0$ & $-9$ & $0$ 
    & 3
    & \cellcolor{lightred}$-1$
    & \cellcolor{lightgreen}$-1/2$
    & \cellcolor{lightgreen}$-0.31$ & $0$  \\

    $W_3$ & $0$ & $0$ & $0$ & $1$ & $-9$
    & 4
    & \cellcolor{lightgreen}$-9/4$
    & \cellcolor{lightgreen}$-3/2$%
    & \cellcolor{lightgreen}$-1/2$ & $0$  \\

    $W_4$ & $4$ & $0$ & $-1$ & $0$ & $0$
    & 2
    & \cellcolor{lightgreen}$-1/2$
    & \cellcolor{lightorange}$-1$
    & \cellcolor{lightgreen}$-1/2$ & $0$  \\

    $W_5$ & $0$ & $0$ & $1$ & $0$ & $-4$
    & 4
    & \cellcolor{lightgreen}$-1$
    & \cellcolor{lightgreen}$-1/3$%
    & \cellcolor{lightgreen}$-1/6$ & $0$  \\
    $W_6-\epsilon_6\mathds{1}$ & $8$ & $0$ & $-4$ & $1$ & $-1$
    & $4$
    & \cellcolor{lightgreen}$-2$
    & \cellcolor{lightgreen}$-4$%
    & \cellcolor{lightgreen}$-2$ & \cellcolor{blue!25}$-4/7$ \\
    $W_7-\epsilon_7\mathds{1}$ & $12$ & $-1$ & $-3$ & $1$ & $0$
    & $3$
    & \cellcolor{lightgreen}$-2$
    & \cellcolor{lightgreen}$-3$%
    & \cellcolor{lightgreen}$-2$ & \cellcolor{blue!25}$-3/20$  \\
    $W_8-\epsilon_8\mathds{1}$ & $0$ & $2$ & $-3$ & $-1$ & $3$
    & $4$
    & \cellcolor{lightgreen}$-1$
    & \cellcolor{lightgreen}$-3$%
    & \cellcolor{lightgreen}$-1/2$ & \cellcolor{blue!25}$-9/86$ \\
    \hline
\end{tabular}
    \caption{{\bf Numerical characterization of four-party entanglement witnesses $W_i$} with unitary and permutation symmetry, defined by a linear combination of the projections $p_\lambda$ with the given coefficients. Witnesses $W_1$ to $W_5$ arise from immanant inequalities~\cite{MaassenSlides}. Witnesses $W_6$ to $W_9$ are to our knowledge new and give rise to previously unknown matrix inequalities. The coefficients $\alpha_{\kappa}$ are the minimum expectation values over states with separability partition $\kappa$, and $\alpha_{1^4}^{PPT}$ is the minimal expectation value over PPT states across all possible bipartitions (detection of PPT entanglement is marked in blue). Green cells indicate that only NPT states are detected, orange cells indicate that no NPT states are detected and red cells that $W_i+\alpha_\kappa$ is positive semidefinite. The precisions $\epsilon_6$, $\epsilon_7$ and $\epsilon_8$ and methodology are given in Appendix~\ref{app:charact4partite}.
    } 
    \label{tab:Main4Partite}
\end{table}

For larger number of parties, we approximate the optimization over separable states with the PPT relaxation. This framework allows us to distinguish $\kappa$-separability structures within seven-qubit systems. As a simple example, from the two-rowed partitions immanant inequality $d_{(4,3)}\per(M)\geq\imm_{(4,3)}(M)$ with $d_{(4,3)}=14$, we consider the seven-qubit full-separability witness
\begin{equation}
    W_7 = 14^2 \Pi_{(7)} - \Pi_{(4,3)}\,.
\end{equation}
By minimizing the expectation value $\langle \tr_1(\dyad{0}W_7)\rangle$ (where $W_7$ is trace-normalized) over PPT states across the desired partitions and numerical inspection on the results, we find that the associated witness $W_7'=W_7+\alpha \Pi_{(4,3)}$ with $\alpha \approx 24/\binom{33}{5}$ detects states not to be a convex combination of $\varrho_{1}\otimes\varrho_{2}\otimes\varrho_{3,4,5,6,7}$ and permutations thereof (i.e. $[5|1|1]$-separable); nor a convex combination of $\varrho_{1}\otimes\varrho_{2,3,4}\otimes\varrho_{5,6,7}$ and permutations thereof (i.e. $[3|3|1]$-separable).

\section{Many-body detection}
In many-body systems it is often particularly important to ensure that no single local system is separable from the rest. Certifying that this is not the case ensures that all local parties have access to multipartite teleportation~\cite{Krlsson1998MultiTele}, distributed quantum computing~\cite{OneWayQC_Briegel2001}, conference key agreement~\cite{Murta_2020ConfKeyAg} and quantum secret sharing~\cite{Hillery1999QSS}. Formally, local users have access to these schemes only if the global $n$-partite state can not be written as
\begin{equation}\label{eq:SemiSep}
   \varrho = \sum_ip_i\varrho_i\otimes\varrho_{1...n\setminus i}.
\end{equation}
States that decompose in this way are known as {\em semi-separable}~\cite{ReviewQEnt_Horo2009}, and specific methods to detect states outside of their set have been explored for certain system sizes~\cite{Kaszlikowski2008SemiSepWitsOptU,Lancien2015SemiSepSDP}. 
Here we use the symmetries of the criteria introduced, to find witnesses for states that are not of the form~\eqref{eq:SemiSep} for any local dimension and number of parties. For that we find analytically the exact value of $\alpha_{n-1|1}$:

\begin{theorem}\label{thm:exact_(1,n-1)_sep}
Let $W = \sum_{\lambda \vdash_d n} c_\lambda \Pi_\lambda$ with $c_\lambda\in\R$. The minimum expectation value of $W$ over $[n-1|1]$-separable states is given by
\begin{equation}
    \alpha_{n-1|1}
    =
    \min_{\substack{ \mu \vdash_d n-1 \\ \nu \sqsubseteq \mu}}
    \sum_{e_k \in \mathrm{AC}_d(\mu)} c_{\mu + e_k} \frac{\prod_{i=1}^{d-1} \nu_i - i - 1 -\mu_k + k}{\prod_{i = 1, i \neq k}^{d} \mu_i - i - \mu_k + k},
\end{equation}
where $\mathrm{AC}_d(\mu)$ is the set of addable corners $e_k$ (with row numbers denoted $k$) of the Young diagram $\mu$ such that $\ell(\mu + e_k) \leq d$, and the minimization is done over partitions $\mu=(\mu_1,\dotsc,\mu_{d})$ and $\nu=(\nu_1,\dotsc,\nu_{d-1})$ satisfying the interlacing condition $\nu \sqsubseteq \mu$, i.e. $\mu_1 \geq \nu_1 \geq \mu_2 \geq \nu_2 \geq \dotsc \geq \mu_{d-1} \geq \nu_{d-1} \geq \mu_d$.
\end{theorem}
The proof is given in Appendix~\ref{app:ThmExact}. 
This theorem provides a recipe to construct simple witnesses to detect many-body states with $\kappa\neq[n-1|1]$. A simple example is the following witness for arbitrary even number of qubits $n$:
\begin{equation}\label{eq:NQubitWitness}
    p_{(n,0)}-\frac{p_{(n/2,n/2)}}{D^2} + \frac{1}{2D^2} < 0 \implies\kappa(\varrho)\not\in[n-1|1]
\end{equation}
where $D := d_{(n/2,n/2)} = \binom{n}{n/2}-\binom{n}{n/2-1}$. This is because the matrix inequality $\per(M)-\imm_{(n-r,r)}(M)/d_{(n-r,r)}\geq 0$ holds for any positive semidefinite matrix $M$~\cite{pate1994immanantRank2,wanless2022liebImInqs}. In the Hilbert space representation, this inequality implies in particular that if $\varrho$ is fully separable, then $p_{(n,0)}-p_{(n/2,n/2)}/D^2\geq 0$. Theorem~\ref{thm:exact_(1,n-1)_sep} shows that for this witness, the smallest expectation value over $[n-1|1]$-separable states is $\alpha_{[n-1|1]}=-1/2D^2$ for qubits (namely with local dimension $d=2$), and therefore by adding the shift $1/2D^2$ Eq.~\eqref{eq:NQubitWitness} holds. The exact derivation of this witness, more general families, dimension-free witnesses, and families of witnesses for high-dimensional systems are given in Appendix~\ref{app:ThmExact}. Concerning the families of witnesses presented in this work, we emphasize that if the permanent-dominance conjecture is shown to be true (namely that the permanent is the largest normalized immanant on positive semidefinite matrices~\cite{LiebPermanent1966,wanless2022liebImInqs}), then any expression of the form $p_{(n)}-p_{\lambda}/d_\lambda^2$ is an entanglement witness and it can be lifted to detect $\kappa$-inseparability with Theorem~\ref{thm:exact_(1,n-1)_sep}. 

Although our main motivation concerning many-body systems is to detect states where no single subsystem is separable from the rest, Theorem~\ref{thm:exact_(1,n-1)_sep} can also be used for other bipartitions. For instance, if the system in hand is composed of an even number of subsystems $n=2t$ of dimension $d$ each, one can group them into $t$ subsystems of dimension $d^2$ each. Then Theorem~\ref{thm:exact_(1,n-1)_sep} provides families of criteria detecting $[n-2|2]$ entanglement in the global system, namely it detects states where no pairs of local subsystems are disentangled from the rest.

\section{Implementation}
\begin{table}[h!]
\centering
\begin{tabular}{l|c|c|c|c}
\hline
Registers & \multicolumn{2}{c|}{$\lambda$, $M_\lambda$, $T_\lambda$} & \multicolumn{2}{c}{$\lambda$, $T_\lambda$} \\ \hline
Method & \cite{EfficientQCirQST_Bacon2006,nguyen2023mixed,grinko2023gelfandtesetlin,burchardt2025hdschur} & \cite{burchardt2025hdschur} & \cite{Krovi2019efficienthigh, burchardt2025hdschur} & \cite{harrow2005thesis,larocca2025quantum}
\\
Depth & $\widetilde{O}(\min(n^5,nd^4))$ & $\widetilde{O}(n^{3.5})$ & $\widetilde{O}(n^{3.5})$ & $\widetilde{O}(n^{4})$ \\
Space & $\widetilde{O}(\min(n^2,d^2))$ & $\widetilde{O}(n^2)$ & $\widetilde{O}(n^2)$ & $\widetilde{O}(n^2)$ \\
\hline
\end{tabular}
\caption{{\bf Scaling of different methods for the weak Schur sampling for general $d$}. Here $\tilde{O}$ denotes polylogarithmic scaling. The left two columns implement the full quantum Schur Transform, which outputs all three registers $\lambda$, $M_\lambda$ and $T_\lambda$. The two known methods for this are BCH algorithm~\cite{EfficientQCirQST_Bacon2006,burchardt2025hdschur} and revised Krovi's algorithm~\cite{burchardt2025hdschur}. The two right columns avoid redundant information about $M_\lambda$ and obtain only the registers $\lambda$ and $T_\lambda$, and thus allow for simpler circuits: these are the original Krovi algorithm~\cite{Krovi2019efficienthigh} and Generalized Phase Estimation~\cite{harrow2005thesis}. Both Krovi and GPE approaches are based on Quantum Fourier Transform for the symmetric group. Finally, a recent weak Schur sampling method for qubits \cite{brahmachari2025optimalqubitpurificationunitary} achieves $\widetilde{O}(n)$ scaling.}
\label{tab:Scaling}
\end{table}
For numerical use of the witnesses found in this work, we provide a symbolic-algebraic SageMath code to generate the Young projectors $\Pi_\lambda$ onto the irreducible subspaces $\mathcal U_\lambda \otimes \mathcal S_\lambda$ of $\cdn$ in a .gz file, ready for load and use in Python language~\cite{RicoGithub}. These can be effectively generated and diagonalized to find $\alpha$-values for nontrivial system sizes. Experimentally, the entanglement criteria derived in this work rely on performing projective measurements onto Young projectors. This is a special case of generalized phase estimation known as {\em weak Schur sampling}, and can be done efficiently using a variety of techniques~\cite{EfficientQCirQST_Bacon2006,harrow2005thesis,Krovi2019efficienthigh,cervero2023weakschursamplinglogarithmic,grinko2023gelfandtesetlin,cervero2024memory,burchardt2025hdschur,brahmachari2025optimalqubitpurificationunitary}. In standard methods, the key idea is to diagonalize the Hilbert space through the unitary Schur transform $U_{\text{ST}}$ as
\begin{equation}
    U_{\text{ST}}\ket{i_1}
    ...
    \ket{i_n} = \sum_{\lambda, M_\lambda, T_\lambda} c_{\lambda,M_\lambda,T_\lambda}\ket{\lambda}
    \ket{M_\lambda}
    \ket{T_\lambda}
\end{equation}
where $c_{\lambda,M_\lambda,T_\lambda}$ are products of the Clebsch-Gordan coefficients and $\lambda\vdash_d n$, $\ket{M_\lambda}\in\mathcal{U}_\lambda$ and $\ket{T_\lambda}\in\mathcal{S}_\lambda$, which can be done efficiently on quantum circuits~\cite{EfficientQCirQST_Bacon2006,harrow2005thesis,Krovi2019efficienthigh,burchardt2025hdschur}. Then the projection $\tr(\Pi_\lambda\varrho)$ of a state $\varrho$ onto the irreducible subspace $\mathcal{U}_\lambda\otimes\mathcal{S}_\lambda$ is obtained by measuring the first register $\ket{\lambda}$ onto the computational basis of $\C^{p(n,d)}$, where $p(n,d)$ is the number of possible partitions $\lambda \vdash_d n$. Since we are interested in obtaining the label $\lambda$, this procedure can be further optimized avoiding a full diagonalization of the Hilbert space (see Table~\ref{tab:Scaling}).

\section{New matrix inequalities}
Up to date, several works (including this contribution) have focused on using immanant inequalities to detect entanglement~\cite{MaassenSlides,Huber2021MatrixFO,long,TP_Rico24}. However, immanant inequalities are an intensively studied field of mathematics and interesting by its own \cite{LiebPermanent1966,haiman1993hecke,pate1992immanantHook,pate1999tensorPerDomN13}. Here we use the correspondence in the converse direction to find new such inequalities.
For that we need a following two-sided formulation of the results of \cite{MaassenSlides} to both directions and rank constraints:
\begin{proposition}\label{prop:WitnessImmanant}
    Let $\{\imm_\lambda\}_\lambda$ be immanants and let $\{\Pi_\lambda\}_\lambda$ be Young projectors. The
    inequality
    \begin{equation}\label{eq:ImInProp}
        \sum_{\lambda\vdash n}a_\lambda\imm_\lambda(G)\geq 0
    \end{equation}
    holds for all $n\times n$ positive semidefinite matrices $G$ of rank $r$, if and only if
    \begin{equation}\label{eq:YoungBPProp}
        \sum_{\lambda\vdash n}\frac{a_\lambda}{d_\lambda}\tr(\Pi_\lambda\varrho)\geq 0
    \end{equation}
    for all $n$-partite separable states in local dimension $r$.
\end{proposition}

This provides a mapping from the space $\C^n\times\C^n$ of $n\times n$ matrices to the $n$-qudit Hilbert space $\cdn$ of some homogeneous local dimension $d$, and vice versa. To our knowledge, here we use this relation by first time to we find new $n$-qudit entanglement witnesses, which provide new matrix inequalities. As an example, the only three-partite witnesses known from immanant inequalities were $\Pi_{(3)}-\Pi_{(2,1)}/4$, $\Pi_{(2,1)} / 4-\Pi_{(1,1,1)}$ and linear combinations thereof, which we have shown to be decomposable. Yet, we found a linearly independent witness $3\Pi_{(3)}-\Pi_{(2,1)}+\Pi_{(1,1,1)}$. Besides detecting bound entanglement, this additional witness provides the following inequality for positive semidefinite $3\times 3$ matrices,
\begin{equation}
    3\per(G) - 2\imm_{21}(G) + \det(G) \geq 0\,.
\end{equation}
Since these three witnesses fully characterize the set of three-qudit block-positive operators given as linear combinations of projectors $\Pi_\lambda$, it follows from Proposition~\ref{prop:WitnessImmanant} that the three corresponding inequalities constitute a complete, analytical characterization for $3\times 3$ matrices. Numerically, we obtain a similar characterization for $4\times 4$ matrices arising from Table~\ref{tab:Main4Partite}.

\bigskip 

\section{Conclusions}

We have developed a method to detect separability partitions $\kappa$, by projecting onto invariant subspaces of the $n$-qudit Hilbert space. By exploiting the symmetries of separability structures, we have been able to analytically fully characterize the entanglement of three qudits with unitary and permutation symmetries, and numerically characterize the case of four qudits. We have found and characterized a rich structure of $\kappa$-separability, genuinely multipartite entanglement and bound entanglement for these system sizes, with few real parameters. The symmetry of the method allows us to compute the criteria introduced for larger system sizes. In particular, we provide an analytical criterion to detect many-body states where no single party is separable to the rest, which are of particular interest in distributed quantum protocols. Following from this result, we provide ready-to-use explicit families tailored to explicit local dimensions. The method introduced can be implemented efficiently in a quantum computer through weak Schur sampling. Besides Physics, our results impact the long-studied theory of matrix immanants in mathematics. We explicitly provide new immanant inequalities, characterize their set for matrices of size $3$ and $4$, and introduce rank-dependent inequalities.


\bigskip

{\em Acknowledgements.} 
A. R. acknowledges financial support from Spanish MICIN (projects: PID2022:141283NBI00;139099NBI00) with the support of FEDER funds, the Spanish Goverment with funding from European Union NextGenerationEU (PRTR-C17.I1), the Generalitat de Catalunya, the Ministry for Digital Transformation and of Civil Service of the Spanish Government through the QUANTUM ENIA project -Quantum Spain Project- through the Recovery, Transformation and Resilience Plan NextGeneration EU within the framework of the Digital Spain 2026 Agenda. D. G. is supported by NWO grant NGF.1623.23.025 (“Qudits in theory and experiment”).
L. H. Z acknowledges support from the Future Talent Guest Stay programme at Technische Universit\"{a}t Darmstadt, during which time this work was initiated.


\appendix
\begin{widetext}

\section{Schur-Weyl duality for two qubits}
Let us illustrate the key components of this method with a simple example, namely the decomposition of the two-qubit space $\C^2\otimes \C^2$. This four-dimensional space decomposes into the three-dimensional triplet subspace, $U\otimes S_{(2,0)}=\text{span}\{\ket{00},\ket{11},\ket{\psi^+}=(\ket{01}+\ket{10})/\sqrt{2}\}$, and the one-dimensional singlet subspace, $U\otimes S_{(1,1)}=\text{span}\{\ket{\psi^-}=(\ket{01}-\ket{10})/\sqrt{2}\}$. For two qubits, the Schur transform is the unitary matrix mapping the computational basis $\{\ket{00},\ket{11},\ket{10},\ket{01}\}$ to the basis $\{\ket{00},\ket{11},\ket{\psi^+},\ket{\psi^-}\}$. This change of basis allows to perform a projective measurement onto the triplet space, given by a projector $\Pi_{(2,0)}=\dyad{00}+\dyad{11}+\dyad{\psi^+}$, and the singlet space, given by a projector $\Pi_{(1,1)}=\dyad{\psi^-}$.

\section{Characterization of three-partite symmetric witnesses}\label{app:Charact3partite}
In this appendix, we fully characterize all tripartite $\k$-separable witnesses of the form $W = c_{(3)}\Pi_{(3)}+c_{(2,1)}\Pi_{(2,1)}+c_{(1,1,1)}\Pi_{(1,1,1)}$ with $c_\lambda \in \mathbb{R}$.
It will be useful to write the projectors $\Pi_\lambda$ in terms of the $(\mathbb{C}^{d})^{\otimes 3}$ permutation operators $\eta_d(\pi)$ that act on tripartite states as $\eta_d(\pi)(\ket{\psi_1}\otimes\ket{\psi_2}\otimes\ket{\psi_3}) = \ket{\psi_{\pi^{-1}(1)}}\otimes\ket{\psi_{\pi^{-1}(2)}}\otimes\ket{\psi_{\pi^{-1}(3)}}$, where $\pi \in S_3$ is a permutation of three objects.
With respect to these permutation operators, we have
\begin{equation}\label{eq:tripartite-projectors}
\begin{aligned}
    \Pi_{(3)} &= \frac{1}{6}\eta_d(\id) + \frac{1}{6}\eta_d((1,2)) + \frac{1}{6}\eta_d((1,3)) + \frac{1}{6}\eta_d((2,3)) + \frac{1}{6}\eta_d((1,2,3)) + \frac{1}{6}\eta_d((1,3,2)) \\
    \Pi_{(2,1)} &= \frac{2}{3}\eta_d(\id) - \frac{1}{3}\eta_d((1,2,3)) - \frac{1}{3}\eta_d((1,3,2)) \\
    \Pi_{(1,1,1)} &= \frac{1}{6}\eta_d(\id) - \frac{1}{6}\eta_d((1,2)) - \frac{1}{6}\eta_d((1,3)) - \frac{1}{6}\eta_d((2,3)) + \frac{1}{6}\eta_d((1,2,3)) + \frac{1}{6}\eta_d((1,3,2)).
\end{aligned}
\end{equation}
In what follows, we will derive the set $\mathcal{W}_{\k} = \{W : \forall\varrho \in \kSEP : \tr(\varrho W) \geq 0 \}$ of $W$ that are positive on all $\k$-separable states.
Notice that $\mathcal{W}_{\k}$ is a convex cone, since it can be easily verified that $\forall p,q \geq 0 \land W,W'\in\mathcal{W}_{\k} : p W + qW' \in \mathcal{W}_{\k}$. 
Non-positive-semidefinite members of this set are witnesses of $k$-separability.

\begin{figure}[h]
    \includegraphics[width=.75\columnwidth]{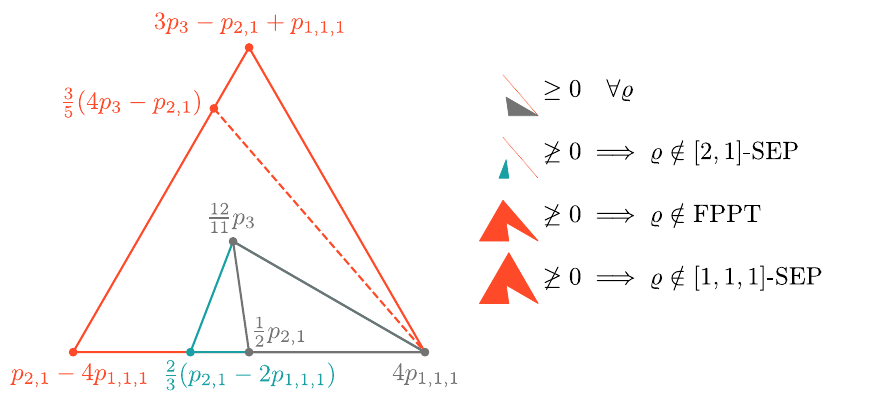}
    
    \caption{
    \label{fig:WitSimp}
    We fully characterize all tripartite witnesses of the form $\sum_{\lambda \vdash_d n} c_\lambda p_\lambda < 0 \implies \varrho \notin \kSEP$.
    Up to normalization, they are all given by $q_1(3p_3-p_{2,1}+p_{1,1,1}) + q_2 (p_{2,1}-4p_{1,1,1}) + q_3 4p_{(1,1,1)}$ for $q_k \geq 0$, plotted here as a simplex.
    A negative value in the corresponding region implies $\k$-inseparability.
    Here, the dashed triangular region formed by the vertices $p_{2,1}-4p_{1,1,1}$, $\propto 4p_3 - p_{2,1}$, and $\propto p_{1,1,1}$ are known from immanant inequalities, but they do not detect bound entanglement, i.e., they are positive for all fully positive partial transpose (PPT) states $\varrho : \varrho^{T_1},\varrho^{T_2},\varrho^{T_3} \geq 0$.
    }
\end{figure}

\subsection{$[2,1]$-separability and genuine multipartite entanglement}
For tripartite systems, states inseparable over the $[2,1]$ partition are genuinely multipartite entangled (GME) in that they cannot be written as a mixture of biseparable states, since $[2,1]$ is the only nontrivial bipartition of three parties.
Using $W = U^{\dag \otimes 3} W U^{\otimes 3}$ for all unitaries $U$,
\begin{equation}
\begin{aligned}
    \alpha_{[2,1]} &= \min_{\ket{\psi_1},\ket{\psi_{23}}}
    \left({\bra{\psi_1}\otimes\bra{\psi_{23}}}\right)
    W
    \left({\ket{\psi_1}\otimes\ket{\psi_{23}}}\right) \\
    &= \min_{\ket{\psi_1}}\min_{\ket{\psi_{23}}}
    \bra{\psi_{23}}U^{\dagger\otimes2}
    \left({\bra{\psi_1}U^\dag\otimes\one\otimes\one}\right)
    W
    \left({U\ket{\psi_1}\otimes\one\otimes\one}\right)
    U^{\otimes2}\ket{\psi_{23}} \\
    &= \min_{\ket{\psi_{23}}}
    \bra{\psi_{23}}
    \left({\bra{1}\otimes\one\otimes\one}\right)
    W
    \underbrace{\left({\ket{1}\otimes\one\otimes\one}\right)}_{\eqqcolon \mathbb{P}_1}
    \ket{\psi_{23}} \\
    &= \operatorname{mineig}[\mathbb{P}_1^\dag W \mathbb{P}_1],
\end{aligned}
\end{equation}
where in the second step we chose $U$ such that $U\ket{\psi_1} = \ket{1}$ for a chosen orthonomal basis $\{\ket{n}\}_{n=1}^d$.

By definition of $\mathbb{P}_1$, $\mathbb{P}_1(\ket{n}\otimes A \otimes B) = \delta_{n,1}A \otimes B$. Consider now the action of $\mathbb{P}_1$ on each $\eta_d(\pi)$. First, it should be obvious that $\mathbb{P}_1^\dag \eta_d(\id) \mathbb{P}_1 = \one\otimes\one$. For the rest,
\begin{equation}\label{eq:GME-projected-witness}
\begin{aligned}
    \mathbb{P}_1^\dag \eta_d((1,2)) \mathbb{P}_1 &= \mathbb{P}_1^\dag \eta_d((1,2)) \sum_{n=1}^d \ketbra{1n}{n} \otimes \one
    = \mathbb{P}_1^\dag \sum_{n=1}^d \ketbra{n1}{n} \otimes \one =  \dyad{1}\otimes\one, \\
    \mathbb{P}_1^\dag \eta_d((1,3)) \mathbb{P}_1 &= \mathbb{P}_1^\dag \eta_d((1,2)) \sum_{n=1}^d \ket{1} \otimes \one \otimes \dyad{n}
    = \mathbb{P}_1^\dag \sum_{n=1}^d \ket{n} \otimes \one \otimes \ketbra{1}{n} =  \one\otimes\dyad{1}, \\
    \mathbb{P}_1^\dag \eta_d((2,3)) \mathbb{P}_1 &= \mathbb{P}_1^\dag \ket{1} \otimes \underbrace{\eta_d((2,3))}_{\hspace{-3em}\text{bipartite swap }S\hspace{-3em}}
    = S, \\
    \mathbb{P}_1^\dag \eta_d((1,2,3)) \mathbb{P}_1 &=
    \mathbb{P}_1^\dag  \eta_d((1,2,3)) \sum_{n,m=1}^d \ketbra{1nm}{nm} = \mathbb{P}_1^\dag \sum_{n,m=1}^d \ketbra{m1n}{nm}
    = \sum_{n=1}^d \ketbra{1n}{n1} = S\left(\one\otimes\dyad{1}\right),\\
    \mathbb{P}_1^\dag \eta_d((1,3,2)) \mathbb{P}_1 &=
    \mathbb{P}_1^\dag  \eta_d((1,3,2)) \sum_{n,m=1}^d \ketbra{1nm}{nm} = \mathbb{P}_1^\dag \sum_{n,m=1}^d \ketbra{nm1}{nm}
    = \sum_{m=1}^d \ketbra{m1}{1m} = S\left(\dyad{1}\otimes \one\right).
\end{aligned}
\end{equation}
Next, let $Q \coloneqq \dyad{1}\otimes\one + \one\otimes\dyad{1}$. Then,
\begin{equation}
\begin{aligned}
    \mathbb{P}_1^\dag \left(\eta_d((1,2,3)) + \eta_d((1,3,2))\right) \mathbb{P}_1 &= SQ, \\
    \mathbb{P}_1^\dag \left(\eta_d((1,2)) + \eta_d((1,3))  + \eta_d((2,3))\right) \mathbb{P}_1 &= Q + S,
\end{aligned}
\end{equation}
with which we have $6\mathbb{P}_1^\dag\Pi_{(3)}\mathbb{P}_1 = \one\otimes\one + S+Q+SQ$, $6\mathbb{P}_1^\dag\Pi_{(2,1)}\mathbb{P}_1 = 4\one\otimes\one - 2SQ$, and $6\mathbb{P}_1^\dag\Pi_{(1,1,1)}\mathbb{P}_1 = \one\otimes\one - S - Q+SQ$.

Finally, from $SQS = \one\otimes\dyad{1} + \dyad{1}\otimes\one = Q$, we can see that $S$ and $Q$ commutes and thus can be simultaneously diagonalised, and we simply guess that the eigenstates are $\{\ket{nn}\}_{n=1}^d$ and $\{\ket{\pm_{nm}} \coloneqq (\ket{nm} \pm \ket{mn})/\sqrt{2}\}_{n=1,m=n+1}^{d-1,d}$, for which there are $\sum_{n'=1}^d + \sum_{n=1}^{d-1}\sum_{m=n+1}^d = d^2$ orthogonal eigenstates, so those are the complete set of eigenstates.
We can check directly that they are eigenstates of $S$ and $Q$ as
\begin{equation}
\begin{aligned}
    S\ket{nn} &= \ket{nn} & S\ket{\pm_{nm}} &= \pm \ket{\pm_{nm}} \\
    Q\ket{nn} &= 2\delta_{n,1}\ket{nn} & Q\ket{\pm_{nm}} &= \delta_{n,1}\ket{\pm_{nm}}.
\end{aligned}
\end{equation}
For $Q\ket{\pm_{nm}}$, keep in mind that $n\geq1$ and thus $m \geq n+1 \geq 2$, so $\left(\dyad{1}\otimes\one\right)\ket{nm} = \delta_{n,1}\ket{nm}$ but $\left(\dyad{1}\otimes\one\right)\ket{mn} = \delta_{m,1}\ket{mn} = 0$.

With these, we can therefore easily derive that
\begin{equation}
\begin{aligned}
    &\operatorname{mineig}\left[
        c_{(3)}\mathbb{P}_1^\dag\Pi_{(3)}\mathbb{P}_1 +
        c_{(2,1)}\mathbb{P}_1^\dag\Pi_{(2,1)}\mathbb{P}_1 +
        c_{(1,1,1)}\mathbb{P}_1^\dag\Pi_{(1,1,1)}\mathbb{P}_1
    \right] \\
    &\quad{}={}\frac{1}{6}\min\bigg[\Big\{
        c_{(3)} \left(
            1 + 1 + 2\delta_{n,1} + 2\delta_{n,1}
        \right) + c_{(2,1)}\left(
            4 - 4 \delta_{n,1}
        \right) + c_{(1,1,1)} \left(
            1 - 1 - 2\delta_{n,1} + 2\delta_{n,1}
        \right)
    \Big\}_{n=1}^{d} \\
    &\quad\qquad\qquad\qquad {}\cup{}\Big\{
        c_{(3)} \left(
            1 \pm 1 + \delta_{n,1} \pm \delta_{n,1}
        \right) + c_{(2,1)}\left(
            4 \mp 2 \delta_{n,1}
        \right) + c_{(1,1,1)} \left(
            1 \mp 1 - \delta_{n,1} \pm \delta_{n,1}
        \right)
    \Big\}_{n=1}^{d-1}\bigg] \\
    &\quad{}={}\frac{1}{6}\min\bigg[\Big\{
        2c_{(3)} \left(
            1 + 2\delta_{n,1}
        \right) + 4c_{(2,1)}\left(
            1 - \delta_{n,1}
        \right)
    \Big\}_{n=1}^{d}\cup\Big\{
        2c_{(3)} \left(
            1 + \delta_{n,1}
        \right) + 2c_{(2,1)}\left(
            2 - \delta_{n,1}
        \right) \Big\}_{n=1}^{d-1} \\
    &\quad\qquad\qquad\qquad {}\cup{}\Big\{
        2c_{(2,1)}\left(
            2 + \delta_{n,1}
        \right) + 2c_{(1,1,1)} \left(
            1 - \delta_{n,1}
        \right)
    \Big\}_{n=1}^{d-1}\bigg] \\
    &\quad{}={}\frac{1}{6}\min\bigg[
        \Big\{
        \underbrace{\min\{6c_{(3)},2c_{(3)} + 4c_{(2,1)}\} \cup
        \min\{4c_{(3)}+2c_{(2,1)},6c_{(2,1)}\}}_{ = 6\min\{c_{(3)},c_{(2,1)}\}}
        \Big\}
    \cup
        \Big\{
            \underbrace{\min\{
                2c_{(3)} + 4c_{(2,1)},
                4c_{(2,1)} + 2c_{(1,1,1)}
            \}}_{\text{if }d>2\text{; first term also appears in front}}
        \Big\}
    \bigg] \\
    &\quad{}={} \min\bigg\{
        c_{(3)},c_{(2,1)},
        \underbrace{
            \frac{2}{3}c_{(2,1)} + \frac{1}{3}c_{(1,1,1)}
        }_{\text{if }d>2}
    \bigg\},
\end{aligned}
\end{equation}
which completes our proof.

Recall that we are interested in the set $\mathcal{W}_{\k} = \{W : \forall \varrho \in \kSEP : \tr(\varrho W) \geq 0 \} = \{W : \alpha_{\k} \geq 0 \}$ that is positive on all $\k$-separable states.
From our expressions for $\alpha_{[2,1]}$, this is $W \in \mathcal{W}_{[2,1]} \iff \alpha_{[2,1]} \geq 0 \iff \{
    c_{(3)} \geq 0 \land c_{(2,1)} \geq 0 \land 2c_{(2,1)} + c_{(1,1,1)} \geq 0
\}$.
Hence, $\mathcal{W}_3$ is the intersection of the cones $c_{(3)} \geq 0$, $c_{(2,1)} \geq 0$, and $2c_{(2,1)} + c_{(1,1,1)} \geq 0$. The extremal rays occur at
\begin{equation}
\begin{gathered}
    0 = c_{(3)} = c_{(2,1)} \leq 2c_{(2,1)} + c_{(1,1,1)} \\
    0 = c_{(2,1)} = 2c_{(2,1)} + c_{(1,1,1)} \leq c_{(3)} \\
    0 = c_{(3)} = 2c_{(2,1)} + c_{(1,1,1)} \leq c_{(2,1)}.
\end{gathered}
\end{equation}
The first condition gives the ray $c_{(1,1,1)}\Pi_{111} : c_{(1,1,1)} \geq 0$, the second gives $c_{(3)}\Pi_3 : c_{(3)} \geq 0$, while the third gives $c_{(2,1)}(\Pi_{21} - 2\Pi_{111}): c_{(2,1)} \geq 0$. Therefore, $\mathcal{W}_{[2,1]} = \mathcal{C}[\{\Pi_{(3)},\Pi_{(1,1,1)},\Pi_{(2,1)}-2\Pi_{(1,1,1)}\}]$, where $\mathcal{C}$ is the convex cone generated from all positive convex combinations of its arguments.

\subsection{$[1,1,1]$-separability or full separability}
We start in a similar manner to the GME bound. By Observation~\ref{obs:KSepAlpha} and $\forall U : W = U^{\dag \otimes 3} W U^{\otimes 3}$, we have that
\begin{equation}
\begin{aligned}
    \alpha_{[1,1,1]} &= \min_{\ket{\psi_1},\ket{\psi_{2}},\ket{\psi_{3}}}
    \left({\bra{\psi_1}\otimes\bra{\psi_{2}}\otimes\bra{\psi_{3}}}\right)
    W
    \left({\ket{\psi_1}\otimes\ket{\psi_{2}}\otimes\ket{\psi_{3}}}\right) \\
    &= \min_{\ket{\psi_1},\ket{\psi_{2}},\ket{\psi_{3}}}
    \left(\bra{\psi_{2}}U^\dagger\otimes\bra{\psi_{3}}U^\dagger\right)
    \left({\bra{\psi_1}U^\dag\otimes\one\otimes\one}\right)
    W
    \left({U\ket{\psi_1}\otimes\one\otimes\one}\right)
    \left(U\ket{\psi_2}\otimes U\ket{\psi_3}\right) \\
    &= \min_{\ket{\psi_{2}},\ket{\psi_{3}}}
    \left({\bra{\psi_{2}}\otimes\bra{\psi_{3}}}\right)
    \left({\bra{1}\otimes\one\otimes\one}\right)
    W
    \underbrace{\left({\ket{1}\otimes\one\otimes\one}\right)}_{\eqqcolon \mathbb{P}_1}
    \left(\ket{\psi_2}\otimes\ket{\psi_3}\right) \\
    &= \min_{\ket{\psi_{2}},\ket{\psi_{3}}}
    \left({\bra{\psi_{2}}V^\dagger\otimes\bra{\psi_{3}}}V^\dagger\right)
    \mathbb{P}_1^\dag
    \left(V^\dagger\otimes\one\otimes\one\right)
    W
    \left(V\otimes\one\otimes\one\right)
    \mathbb{P}_1
    \left(V\ket{\psi_2}\otimes V\ket{\psi_3}\right)
\end{aligned}
\end{equation}
Now, choose the unitary $V$ that rotates the two-dimensional subspace spanned by $\{\ket{1},\ket{\psi_2}\}$ to the two-dimensional subspace spanned by $\{\ket{1},\ket{2}\}$, such that $V\ket{1} = \ket{1}$ and $V\ket{\psi_2} = \sqrt{p}\ket{1} + \sqrt{1-p}\ket{2}$ for some $0 \leq p \leq 1$. As such,
\begin{equation}
\begin{aligned}
    \alpha_{[1,1,1]}
    &= \min_{\ket{\psi_{3}}}\min_{0\leq p \leq 1}
    \Big[\left(\sqrt{p}\bra{1} + \sqrt{1-p}\bra{2}\right)\otimes\bra{\psi_3}\Big]
    \mathbb{P}_1^\dag W \mathbb{P}_1
    \Big[\underbrace{\left(\sqrt{p}\ket{1} + \sqrt{1-p}\ket{2}\right)}_{\coloneqq\ket{\phi_p}}\otimes\ket{\psi_3}\Big] \\
    &= \min_{0\leq p \leq 1}\min_{\ket{\psi_{3}}}
    \bra{\psi_3}\left(\bra{\phi_p}\otimes\one\right)
    \mathbb{P}_1^\dag W \mathbb{P}_1
    \underbrace{\left(\ket{\phi_p}\otimes\one\right)}_{\coloneqq\mathbb{P}_p}\ket{\psi_3} \\
    &= \min_{0\leq p \leq 1}\operatorname{mineig}\left[\mathbb{P}_p^\dag \mathbb{P}_1^\dag W \mathbb{P}_1 \mathbb{P}_p\right].
\end{aligned}
\end{equation}
We clearly have $\mathbb{P}_p^\dag\mathbb{P}_1^\dag \eta_d(\id) \mathbb{P}_1\mathbb{P}_p = \one$, and acting $\mathbb{P}_p$ upon Eq.~\eqref{eq:GME-projected-witness} gives
\begin{equation}
\begin{aligned}
    \mathbb{P}_p^\dag\mathbb{P}_1^\dag \eta_d((1,2)) \mathbb{P}_1\mathbb{P}_p &= p\one, \\
    \mathbb{P}_p^\dag\mathbb{P}_1^\dag \eta_d((1,3)) \mathbb{P}_1\mathbb{P}_p &= \sum_{n=1}^d \mathbb{P}_p^\dagger\left(
        \one\otimes\ketbra{1}{1}
    \right)\ket{\phi_p}\otimes\ketbra{n}{n} = \ketbra{1}{1}, \\
    \mathbb{P}_p^\dag\mathbb{P}_1^\dag \eta_d((2,3)) \mathbb{P}_1\mathbb{P}_p &= \sum_{n=1}^d \mathbb{P}_p^\dagger\ket{n}\otimes\ketbra{\phi_p}{n} = \ketbra{\phi_p}{\phi_p}, \\
    \mathbb{P}_p^\dag\mathbb{P}_1^\dag \eta_d((1,2,3)) \mathbb{P}_1\mathbb{P}_p &= \mathbb{P}_p^\dagger S\ket{\phi_p}\otimes\ketbra{1}{1} = \sqrt{p}\ketbra{\phi_p}{1},\\
    \mathbb{P}_p^\dag\mathbb{P}_1^\dag \eta_d((1,3,2)) \mathbb{P}_1\mathbb{P}_p &= \sqrt{p}\mathbb{P}_p^\dagger S \ket{1}\otimes\one =
    \sqrt{p}\ketbra{1}{\phi_p}.
\end{aligned}
\end{equation}
From this, we find that
\begin{equation}
\begin{aligned}
    \mathbb{P}_{p}^\dagger\mathbb{P}_{1}^\dagger \Pi_{3} \mathbb{P}_{1}\mathbb{P}_p &= \frac{1}{6}\left[
        (1+p)\one + \ketbra{1}{1} + \ketbra{\phi_p}{\phi_p} + \sqrt{p}\left( \ketbra{\phi_p}{1} + \ketbra{1}{\phi_p} \right)
    \right] \\
    \mathbb{P}_{p}^\dagger\mathbb{P}_{1}^\dagger \Pi_{21} \mathbb{P}_{1}\mathbb{P}_p &= \frac{1}{6}\left[
        4\one - 2\sqrt{p}\left( \ketbra{\phi_p}{1} + \ketbra{1}{\phi_p} \right)
    \right] \\
    \mathbb{P}_{p}^\dagger\mathbb{P}_{1}^\dagger \Pi_{111} \mathbb{P}_{1}\mathbb{P}_p &= \frac{1}{6}\left[
        (1-p)\one - \ketbra{1}{1} - \ketbra{\phi_p}{\phi_p} + \sqrt{p}\left( \ketbra{\phi_p}{1} + \ketbra{1}{\phi_p} \right)
    \right]
\end{aligned}
\end{equation}
Further define $\mathbb{P}_2 \coloneqq \ketbra{1}{1} + \ketbra{2}{2}$ and
\begin{equation}
    \sigma_2 \coloneqq \ketbra{\phi_p}{1} + \ketbra{1}{\phi_p} - \sqrt{p}\mathbb{P}_2
    = \sqrt{p}\left(\ketbra{1}{1}-\ketbra{2}{2}\right) + \sqrt{1-p}\left(\ketbra{1}{2} + \ketbra{2}{1}\right).
\end{equation}
It can be easily verified that $\sigma_2$ has the eigenvalues $\pm 1$ with corresponding eigenstates $\ket{\pm} \propto \sqrt{1\pm\sqrt{p}}\ket{1} \pm \sqrt{1\mp\sqrt{p}}\ket{2}$. In terms of $\sigma_2$, we have $\ketbra{\phi_p}{\phi_p} = \ketbra{2}{2} + \sqrt{p}\sigma_2$, with which we can rewrite the other observables as
\begin{equation}
\begin{aligned}
    \mathbb{P}_{p}^\dagger\mathbb{P}_{1}^\dagger \Pi_{3} \mathbb{P}_{1}\mathbb{P}_p &= \frac{1}{6}\left[
        (1+p)\left(\one + \mathbb{P}_2\right) + 2\sqrt{p}\sigma_2
    \right] \\
    \mathbb{P}_{p}^\dagger\mathbb{P}_{1}^\dagger \Pi_{21} \mathbb{P}_{1}\mathbb{P}_p &= \frac{1}{6}\left[
        4\one - 2p\mathbb{P}_2 - 2\sqrt{p}\sigma_2
    \right] \\
    \mathbb{P}_{p}^\dagger\mathbb{P}_{1}^\dagger \Pi_{111} \mathbb{P}_{1}\mathbb{P}_p &= \frac{1}{6}\left[
        (1-p)\left(\one - \mathbb{P}_2\right)
    \right]
\end{aligned}
\end{equation}
Then, we have
\begin{equation}
\begin{aligned}
    \mathbb{P}_p^\dag \mathbb{P}_1^\dag W \mathbb{P}_1 \mathbb{P}_p &=
    \frac{1}{6}\big\{
        \left[ c_{(3)} + 4c_{(2,1)} + c_{(1,1,1)} + (c_{(3)}-c_{(1,1,1)})p \right]\one \\
    &\qquad\qquad{}+{}
        \left[ (c_{(3)}-c_{(1,1,1)}) (c_{(3)}-2c_{(2,1)}+c_{(1,1,1)})p \right]\mathbb{P}_{2} +
        \left[ 2(c_{(3)}-c_{(2,1)})\sqrt{p} \right] \sigma_2.
    \big\}
\end{aligned}
\end{equation}
Since $\ket{\pm}$ and $\ket{n \geq 3}$ are simultaneous eigenstates of $(\one,\mathbb{P}_{2},\sigma_2)$ with eigenvalues $(1,1,\pm 1)$ and $(1,0,0)$ respectively, we find that
\begin{equation}
    \operatorname{mineig}\left[\mathbb{P}_p^\dag \mathbb{P}_1^\dag W \mathbb{P}_1 \mathbb{P}_p\right] =
    \frac{1}{6}\min\Bigg\{
        2c_{(3)}+4c_{(2,1)} + 2(c_{(3)}-c_{(2,1)})p - 2 \left\lvert c_{(3)}-c_{(2,1)} \right\rvert \sqrt{p},
        \underbrace{c_{(3)}+4c_{(2,1)}+c_{(1,1,1)} + (c_{(3)}-c_{(1,1,1)})p}_{\text{if $d > 2$}}
    \Bigg\},
\end{equation}
and finally noting that the expressions are at most quadratic in $\sqrt{p}$, for which the minimum value over the range $0 \leq \sqrt{p} \leq 1$ can be found analytically,
\begin{equation}
\begin{aligned}
    \alpha_{[1,1,1]} &= \min_{0\leq p \leq 1}\operatorname{mineig}\left[\mathbb{P}_p^\dag \mathbb{P}_1^\dag W \mathbb{P}_1 \mathbb{P}_p\right] \\
    &= \frac{1}{6}\min_{0\leq p \leq 1}\min\Bigg\{
        2c_{(3)}+4c_{(2,1)} + 2(c_{(3)}-c_{(2,1)})p - 2 \left\lvert c_{(3)}-c_{(2,1)} \right\rvert,
        \underbrace{c_{(3)}+4c_{(2,1)}+c_{(1,1,1)} + (c_{(3)}-c_{(1,1,1)})p}_{\text{if $d > 2$}}
    \Bigg\}\\
    &= \frac{1}{6}\min\Bigg\{
        \begin{cases}
            6c_{(2,1)} & \text{if $c_{(3)} \leq c_{(2,1)}$} \\
            \frac{3}{2}c_{(3)} + \frac{9}{2}c_{(2,1)} & \text{otherwise}
        \end{cases},
        \underbrace{
            \begin{cases}
                2c_{(3)} + 4c_{(2,1)} & \text{if $c_{(3)} \leq c_{(1,1,1)}$} \\
                c_{(3)} + c_{(1,1,1)} + 4c_{(2,1)} & \text{otherwise}
            \end{cases}
        }_{\text{if $d > 2$}}
    \Bigg\} \\
    &= \min\Bigg\{
        \frac{1}{4}c_{(3)} + \frac{3}{4}\min\{c_{(3)},c_{(2,1)}\},
        \underbrace{
           \frac{2}{3}c_{(2,1)} + \frac{1}{6}c_{(3)} +
           \frac{1}{6}\min\{c_{(3)},c_{(1,1,1)}\}
       }_{\text{if $d > 2$}}
    \Bigg\} \\
    &= \min\Bigg\{
        c_{(3)},
        \frac{1}{4}c_{(3)} + \frac{3}{4}c_{(2,1)},
        \underbrace{
           \frac{2}{3}c_{(2,1)} + \frac{1}{6}c_{(3)} +
           \frac{1}{6}c_{(1,1,1)}
       }_{\text{if $d > 2$}}
    \Bigg\}.
\end{aligned}
\end{equation}
Turning back to the set $\mathcal{W}_{\k} = \{W : \forall \varrho \in \kSEP : \tr(\varrho W) \geq 0 \} = \{W : \alpha_{\k} \geq 0 \}$, we can find that $W \in \mathcal{W}_{[1,1,1]} \iff \alpha_{[1,1,1]} \geq 0 \iff \{
    c_{(3)} \geq 0 \land c_{(3)} + 3c_{(2,1)} \geq 0 \land 4c_{(2,1)} + c_{(3)} + c_{(1,1,1)} \geq 0
\}$.
The extremal rays of $\mathcal{W}_{\k}$ occur at
\begin{equation}
\begin{gathered}
    0 = c_{(3)} = c_{(3)} + 3c_{(2,1)} \leq 4c_{(2,1)} + c_{(3)} + c_{(1,1,1)} \\
    0 =  c_{(3)} + 3c_{(2,1)} = 4c_{(2,1)} + c_{(3)} + c_{(1,1,1)} \leq c_{(3)} \\
    0 = 4c_{(2,1)} + c_{(3)} + c_{(1,1,1)} = c_{(3)} \leq c_{(3)} + 3c_{(2,1)}.
\end{gathered}
\end{equation}
The first condition gives the ray $c_{(1,1,1)}\Pi_{(1,1,1)} : c_{(1,1,1)} \geq 0$, the second gives $c_{(1,1,1)}(3\Pi_{(3)} - \Pi_{(2,1)} + \Pi_{(1,1,1)}): c_{(1,1,1)} \geq 0$, while the third gives $c_{(2,1)}(\Pi_{(2,1)} - 4\Pi_{(1,1,1)}): c_{(2,1)} \geq 0$.
Therefore, $\mathcal{W}_{[1,1,1]} = \mathcal{C}[\{\Pi_{(1,1,1)},\Pi_{(2,1)}-4\Pi_{(1,1,1)},3\Pi_{(3)}-\Pi_{(2,1)}+\Pi_{(1,1,1)}\}]$.

\subsection{Full partial transpose positivity}
As mentioned in the main text, the separability witnesses $W_+ \coloneqq 4\Pi_{(3)} - \Pi_{(2,1)}$ and $W_- \coloneqq \Pi_{(2,1)} - 4\Pi_{(1,1,1)}$ are known from existing immanant inequalities.
Here, we shall show that they are decomposable, and therefore positive $\tr(\varrho W_\pm) \geq 0$ on states $\varrho \in \operatorname{FPPT}$ that are fully positive-partial-transpose such that $\varrho: \varrho^{T_1},\varrho^{T_2},\varrho^{T_3} \geq 0$.

Using \eqref{eq:tripartite-projectors}, we write $W_\pm$ in terms of the permutation operators as
\begin{equation}
\begin{aligned}
    W_\pm &= \frac{2}{3}\left[
        \eta_d((1,2)) + \eta_d((1,3)) + \eta_d((2,3))
    \right] \pm \left[
        \eta_d((1,2,3)) +
        \eta_d((1,3,2))
    \right]\\
    &= \sum_{n=1}^3 \frac{1}{3}\left[
        \eta_d((n,n+1)) + \eta_d((n+1,n+2)) \pm
        \eta_d((1,2,3)) \pm \eta_d((1,3,2))
    \right],\\
    &= \sum_{n=1}^3 \frac{1}{3}\left[
        \eta_d((n,n+1)) + \eta_d((n+1,n+2)) \pm
        \eta_d((n,n+1))\eta_d((n+1,n+2)) \pm
        \eta_d((n+1,n+2))\eta_d((n,n+1))
    \right] \\
    &= \sum_{n=1}^3 \frac{4}{3}\left\{
        \frac{1}{2}\left[\one \pm \eta_d((n+1,n+2))\right]
        \eta_d((n,n+1))
        \frac{1}{2}\left[\one \pm \eta_d((n+1,n+2))\right] 
    \right\},
\end{aligned}
\end{equation}
where the sum and difference in $\eta_d((n,n+1))$ and $\eta_d((n,n+2))$ are to be understood as modular, in that $\eta_d((3,3+1)) = \eta_d((3,1))$ and $\eta_d((3,3+2)) = \eta((3,2))$, and we decomposed the cyclic permutation $(1,2,3)$ into pairwise swaps.

Now, let us recognize that $\one_n\otimes\one_m\otimes\one_l \pm \eta_d((m,l)) \eqqcolon 2\one_n \otimes P_{\pm(m,l)}$, where $n\neq m \neq l$, is the projector $P_{+(m,l)} = \Pi_{(2)}$ (respectively $P_{-(m,l)} = \Pi_{(1,1)}$) onto the symmetric (respectively antisymmetric) subspace of the $m$th and $l$th system.
Furthermore, we also have that the partial transpose of the swap is $\eta_d((n,m))^{T_n} = 
\ket{\Phi}_{n,m}\!\!\bra{\Phi}\otimes\one_l$ where $\ket{\Phi} = \sum_{j=1}^d\ket{jj}$ is the unnormalized EPR state.
Substituting this back into the equation above,
\begin{equation}
\begin{aligned}
    W_\pm
    &= \sum_{n=1}^3 \frac{4}{3}
        \left(\one_{n}\otimes P_{\pm(n+1,n+2)}\right)
        \left(\ket{\Phi}_{n,n+1}\!\!\bra{\Phi}\otimes\one_{n+2}\right)^{T_n}
        \left(\one_{n}\otimes P_{\pm(n+1,n+2)}\right)
    \\
    &= \sum_{n=1}^3 \frac{4}{3}
        \left[
        \left(\one_{n}\otimes P_{\pm(n+1,n+2)}\right)
        \left(\ket{\Phi}_{n,n+1}\!\!\bra{\Phi}\otimes\one_{n+2}\right)
        \left(\one_{n}\otimes P_{\pm(n+1,n+2)}\right)
        \right]^{T_n}.
\end{aligned}
\end{equation}
Finally, notice that $(\one_{n}\otimes P_{\pm(m,l)})
        (\ket{\Phi}_{n,m}\!\!\bra{\Phi}\otimes\one_{l})
        (\one_{n}\otimes P_{\pm(m,l)})$
is proportional to the projection of the projector $\ketbra{\Phi}{\Phi}\otimes\one$ onto the symmetric or antisymmetric subspace.
This is clearly a positive operator, which implies that $W_\pm$ is the sum of partial transpose of positive operators, and therefore a decomposable witness.

\section{Characterization of four-partite systems}\label{app:charact4partite}
We first characterize numerically the set of four-ququarts fully PPT states. Similar results hold for this system size, in the sense that the expectation values $\mathcal{I}_{4,\mathrm{PPT}} = \{\tr \Pi_\lambda\rho|\rho~\mathrm{is~FPPT} \}$ of the fully PPT (FPPT) set are found to form a polytope. To learn this, we sample a number of random extreme points and cluster them into vertices in order to identify the vertices of the resulting polytope. Extreme points are sampled by maximizing random linear functions, which will give an extreme point for almost every linear function~\cite{tardella2011fundamental}. For every face of the polytope, there is an operator $F$ encoding the corresponding hyperplane, and we confirm that for fully PPT (FPPT) $\rho$ , $\tr F\rho\ge 0$ holds, with at least $4$ optimal FPPT states $\tr F\rho= 0$, showing that this empirically determined polytope is identical to $\mathcal{I}_{4,\mathrm{PPT}}$. 
We identify $7$ distinct vertices of this $4$-dimensional polytope. Since we can find integer coordinates for the vertices, we eschew normalization in table~\ref{tab:44fppt}, which lists and classifies the vertices of the FPPT polytope. The dual description in terms of polytope facets can be found in table~\ref{tab:dualdesc}.

\begin{table}[h!]
    \centering
    \begin{tabular}{ccccc|c||ccc}
         $\Pi_4$ & $\Pi_{3,1}$ & $\Pi_{2,2}$ & $\Pi_{2,1,1}$ & $\Pi_{1,1,1,1}$
         & $/\mathcal{N}$ &  Entanglement & Certificate & $\substack{\min \\ W\in wit}\tr(W\varrho)$\\
         1 & 0 & 0 & 0 & 0 & /1 & f.sep. & $\ket{0000}$ & 0 \\
         1 & 3 & 0 & 0 & 0 & /4 & f.sep. & $\ket{0001}$ & 0\\
         1 & 9 & 4 & 0 & 0 & /14 & ent. & $W_1$ & $-4/7$\\
         2 & 9 & 6 & 0 & 0 & /17 & ent. & $W_1$ & $-8/17$\\
         1 & 9 & 4 & 6 & 0 & /20 & ent. & $W_2$ & $-3/20$\\
         5 & 36 & 18 & 27 & 0 & /86 & ent. & $W_3$ & $-9/86$\\
         1 & 9 & 4 & 9 & 1 & /24 & f.sep. & $\ket{0123}$ & 0
    \end{tabular}
    \caption{Rows are integer coordinates
    $\mathcal{N}_i \tr{(\Pi_\lambda\rho_i)}$
    for extreme points $\rho_i$ of the fully PPT polytope with local dimension 4. For every extreme point, we successfully attempt to certify either that it is fully separable or provide a close-to-optimal indecomposable witness for detecting it. Separability, on the other hand, is certified by a product state with expectation values on Young projectors matching those in the given row. If such a product state is found, its twirl under $U^{\otimes 4}$ gives the associated Werner state. 
    } 
    \label{tab:44fppt}
\end{table}

To obtain witnesses, we consider an inner approximation and an outer approximation to the set of fully separable states. First we describe the inner approximation. We consider the set of expectation values of $\Pi_\lambda$ occurring in the fully separable set, i.e. $\mathcal{I}_4$, the convex closure of $\{\bra{abcd}\Pi_\lambda\ket{abcd}\}_{\ket{abcd}}$. We sample a dataset of 
real numbers $f_\lambda$ and minimize $\bra{abcd}\left(\sum_\lambda f_\lambda \Pi_\lambda\right)\ket{abcd}$ over $\ket{abcd}$, with the usual seesaw method~\cite{weinbrenner2025quantifying}. This optimization leads to an extremal point of $\mathcal{I}_4$ if it converges to a global optimum, and yields an interior point of $\mathcal{I}_4$ otherwise. 
Creating a large number of product states $\ket{\psi^{\mathrm{pr}}_i}$ in this manner provides a polytope $\tilde{\mathcal{I}}_4=\mathrm{conv}(\{\bra{\psi^{\mathrm{pr}}_i} \Pi_\lambda\ket{\psi^{\mathrm{pr}}_i}\}_i)$ given by the convex hull of the product states found. This polytope is an inner approximation of $\mathcal{I}_4$. 
Now, consider a point $\sum_\lambda x_\lambda\Pi_\lambda \not\in\mathcal{I}_4$ out of the set of separable states (namely an entangled state). We find a witness $W=\sum_\lambda c_\lambda\Pi_\lambda$ with respect to $\tilde{\mathcal{I}}_4$ detecting it, by 
a linear program $\min_{c_\lambda} \sum_\lambda c_\lambda x_\lambda,~\mathrm{s.t.}~\tr{\sum_\lambda c_\lambda \Pi_\lambda}=1,~\sum_\lambda c_\lambda\xi_\lambda\ge 0 ~(\forall\xi_\lambda\in\tilde{\mathcal{I}}_4)$. Notice that by assumption this witness has nonnegative expectation value on $\tilde{\mathcal{I}}_4$, but in principle might not for $\mathcal{I}_4$. However, it will be close to the optimal witness with respect to $\mathcal{I}_4$ if a sufficiently large number of product states has been sampled.

Then, we can use any relaxation $\mathcal{R}$ of the fully separable set, which is tighter than the FPPT set to determine $\epsilon:=\min_{\rho\in\mathcal{R}}\tr W \rho\le\min_{\ket{abcd}} \bra{abcd}W\ket{abcd}$. 
Then $\epsilon\mathds{1}+W$ is a certified entanglement witness. 
In our case, we use $\epsilon = -\min \tr W (\ketbra{0}{0}_\mathrm{A}\otimes X_{\mathrm{BCD}})$ over a FPPT states $X_{BCD}$ with $\tr(X)=1$. We perform this procedure to detect all extreme points of the PPT set which could not be reduced to separable states with the seesaw algorithm.
For the witnesses which we obtain in this manner, the parameter $\epsilon$ given by this relaxation is relatively small compared to the negative eigenvalues of the witness. 
Thus, we conjecture that it would vanish in an exact treatment of the separable set.
 
\begin{table}[h!]
    \centering
    \begin{tabular}{c|ccccc||cc}
         $\Pi_\lambda$ & $\Pi_4$ & $\Pi_{3,1}$ & $\Pi_{2,2}$ & $\Pi_{2,1,1}$ & $\Pi_{1,1,1,1}$
          & $\epsilon$ & $/\mathcal{N}$ \\
         $w_{1,\lambda}$ & 8 & 0 & -4 & 1 & -1 &   \SI{4.048e-2}{}  &$/164$ \\
         $w_{2,\lambda}$ & 12 & -1 & -3 & 1 & 0 &  \SI{5.476e-3}{} & $/164$\\
         $w_{3,\lambda}$ & 0 & 2 & -3 & -1 & 3 &  \SI{1.759e-2}{} & $/108$
    \end{tabular}
    \caption{Rows define indecomposable witnesses as $W_i = \left(\sum_\lambda (w_{i,\lambda}+\epsilon)\Pi_\lambda\right)/(\mathcal{N}+4^4\epsilon)$, that is, we use the dual of the basis in the previous table~\ref{tab:44fppt}. To our knowledge, the associated immanant inequalities are maximally tight for $4\times 4$ matrices and unknown in the literature.
    } 
    \label{tab:indecomposable_w}
\end{table}

\begin{table}[h!]
    \centering
    \begin{tabular}{c|ccccc||c}
         $\Pi_\lambda$ & $\Pi_4$ & $\Pi_{3,1}$ & $\Pi_{2,2}$ & $\Pi_{2,1,1}$ & $\Pi_{1,1,1,1}$ & $/\mathcal{N}$\\
        $f_{1,\lambda}$ & 0 & 0 & 0 & 0 & 1 & $/1$\\
$f_{2,\lambda}$ & 0 &  0 &  0 &  1 & -9& $/36$\\
$f_{3,\lambda}$ & 0 &   0 &  3 & -2 &   6 & $/36$\\
$f_{4,\lambda}$ & 0 &  6 & -9 & -2 &  0 & $/360$\\
$f_{5,\lambda}$ & 6 & -2 &  3 &  0 &  0 & $/60$\\
$f_{6,\lambda}$ & 18 &  2 & -9 &  0 &  0 & $/540$
    \end{tabular}
    \caption{Rows are integer coordinates for the dual of the PPT polytope with local dimension 4. Hence, this is a list of all decomposable witnesses $F_i = \sum_\lambda f_{i,\lambda}P_\lambda$ expressed in the basis of Young projectors. This is confirmed by minimizing the functional $\tr{F_i\rho}$ over (fully) PPT $\rho$. We further observe that for each of the $F$ in the list, $\tr{F_i\rho_j}$ vanishes for at least $4$ different vertices $\rho_j$ of the PPT polytope, confirming that the PPT set does not extend outside the previously identified PPT polytope, and is tightly contained inside, that is, the two sets are identical. Any of these six witnesses can be turned into an imminant inequality. 
    } 
    \label{tab:dualdesc}
\end{table}

\section{Proof of Theorem \ref{thm:exact_(1,n-1)_sep} and families of witnesses}\label{app:ThmExact}
In this section, we present a proof of Theorem \ref{thm:exact_(1,n-1)_sep}. It is based on Schur--Weyl duality and representation theory of the unitary group, and specifically the Gelfand--Tsetlin basis of the unitary group. We refer a reader to \cite{grinko2025mixed} for the necessary background.

\begin{proof}
Using symmetries of $W = \sum_{\lambda \vdash_d n} c_\lambda \Pi_\lambda$, where $\Pi_\lambda$ is the Young projector onto $\lambda$ isotypic component in the Schur--Weyl duality, the calculation of  $\alpha_{n-1|1}(W)$ can be reduced to the following problem:
\begin{align}
    \alpha_{n-1|1} &= \min_{\rho \in \kSEP} \tr \sof{\rho W} \\
    &= \min_{\ket{\Psi} \in (\C^d)^{\otimes n-1}} \min_{\ket{\psi} \in \C^d}  \tr \sof{ \of{\ketbra{\psi}{\psi} \otimes \ketbra{\Psi}{\Psi}} W} \\
    &= \min_{\ket{\Psi}} \tr \sof{\of{\ketbra{d}{d} \otimes \ketbra{\Psi}{\Psi}} W} \\
    &= \operatorname{mineig} \sum_{\lambda \vdash_d n} c_\lambda X^\lambda,
\end{align}
where for every partition $\lambda \vdash_d n$ we define the operator $X^\lambda \in \mathrm{End}((\C^d)^{\otimes n-1})$ as
\begin{align}
    X^\lambda := \of{\bra{d} \otimes I^{\otimes n-1}} \Pi_\lambda \of{\ket{d} \otimes I^{\otimes n-1}}.
\end{align}
Now note that 
\begin{align}
\label{eq:symmetry_Sn}
    [X^\lambda, \pi] &= 0 \text{ for every } \pi \in S_{n-1} \\
\label{eq:symmetry_Ud-1}
    [X^\lambda, U^{\otimes n-1}] &= 0, \text{ for every } U \in \U{d-1} := \Set{U \in \U{d} \given U\ket{d}=\ket{d}} \subset \U{d},
\end{align}
where the second symmetry follows from
\begin{equation}
    U^{\dagger \otimes n-1} X^\lambda U^{\otimes n-1} = \of{\bra{d} \otimes I^{\otimes n-1}} U^{\dagger \otimes n} \Pi_\lambda U^{\otimes n} \of{\ket{d} \otimes I^{\otimes n-1}} = X^\lambda \, \text{ for every } U \in \U{d-1}.
\end{equation}
Equation~(\ref{eq:symmetry_Sn}) means that, due to Schur--Weyl duality, we have
\begin{equation}
    U_{\text{ST}} X^\lambda U_{\text{ST}}^\dagger = \bigoplus_{\mu \vdash_d n-1} I_{d_\mu} \otimes X^\lambda_\mu,
\end{equation}
where $X^\lambda_\mu$ acts in the Gelfand--Tsetlin basis of the Weyl module $\mu$ and has dimension $m_\mu$. Due to Equation~(\ref{eq:symmetry_Ud-1}) and the subgroup adaptive property of the Gelfand--Tsetlin basis, $X^\lambda_\mu$ is actually a diagonal matrix:
\begin{equation}
    X^\lambda_\mu = \bigoplus_{\nu \sqsubseteq \mu} x^\lambda_{\mu,\nu} I_{m_\nu}
\end{equation}
In fact, $x^\lambda_{\mu,\nu}$ is a square of a particular reduced Wigner coefficient \cite{Vilenkin1992}. More precisely,
\begin{equation}
    x^\lambda_{\mu,\nu} = \delta_{\lambda \in \mu + \square} \frac{\prod_{i=1}^{d-1} \nu_i - i -\lambda_k + k}{\prod_{i = 1, i \neq k}^{d} \mu_i - i - \mu_k + k},
\end{equation}
where if $\lambda \in \mu + \square$ then $k$ is defined via identity $\lambda = (\mu_1,\dotsc,\mu_{k-1},\mu_k+1,\mu_{k+1},\dotsc,\mu_d)$, and symbol $\delta_{\lambda \in \mu + \square}$ is $1$ if Young diagram $\lambda$ is obtained from $\mu$ by adding a box and it is $0$ otherwise.
In summary, the spectrum of $X^\lambda$ is described by the following set:
\begin{equation}
    \mathrm{spec} \, X^\lambda = \Set{x^\lambda_{\lambda-r,\nu} \given r \in \mathrm{RC}(\lambda), \, \nu \sqsubseteq \lambda-r},
\end{equation}
where $\mathrm{RC}(\lambda)$ denotes the set of removable boxes of the Young diagram $\lambda$, and the symbol $\sqsubseteq$ denotes interlacing condition. Now we just need to select the minimum eigenvalue from this explicit spectrum. That proves the claimed result.
\end{proof}

\begin{cor}
For any $k\geq 0$, $d \geq k + 2$, $n \geq k+2$ consider a witness
\begin{equation}
W^{(k)} = \frac{\Pi_{(n-k,1^{k},0^{d-k-1})}}{\binom{n-1}{k}^2} - \frac{\Pi_{(n-k-1,1^{k+1},0^{d-k-2})}}{\binom{n-1}{k+1}^2}.
\end{equation}
Then its minimal expectation value $\alpha_{n-1|1}$ is independent of $d$ and equals
\begin{equation}
\alpha_{n-1|1}(W^{(k)}) \;=\;
\begin{cases}
-\dfrac{1}{\binom{n-1}{k+1}^{2}}, 
& \text{if } k \leq n-3 \\[15pt]
-\dfrac{n-2}{(n-1)n},
& \text{if } k = n-2.
\end{cases}
\end{equation}
\end{cor}

\begin{proof}
In the following, we use the notation
\begin{equation}
\lambda_1=(n-k,1^{k},0^{d-k-1}), 
\qquad 
\lambda_2=(n-k-1,1^{k+1},0^{d-k-2}),
\qquad
d_{\lambda_1}=\binom{n-1}{k},
\quad
d_{\lambda_2}=\binom{n-1}{k+1},
\end{equation}
and expand the witness as 
\begin{equation}
W^{(k)}=\sum_{\lambda}c_\lambda \Pi_\lambda,
\qquad
c_{\lambda_1}=\frac{1}{d_{\lambda_1}^2},
\qquad
c_{\lambda_2}=-\frac{1}{d_{\lambda_2}^2},
\qquad
c_{\lambda}=0\ \text{otherwise}.
\end{equation}
The minimum we are optimizing is
\begin{equation}
\alpha_{n-1|1}(W^{(k)})
=
\min_{\mu\vdash_d (n-1)}\ 
\min_{\nu\sqsubseteq \mu}
\ 
\sum_{e_k\in\mathrm{AC}_d(\mu)} c_{\mu+e_k}\,w_{k}^{(\mu,\nu)},
\end{equation}
where $w_{k}^{(\mu,\nu)}$ is the coefficient from Theorem~\ref{thm:exact_(1,n-1)_sep}:
\begin{equation}
    w_{k}^{(\mu,\nu)} = \frac{\prod_{i=1}^{d-1} \nu_i - i - \mu_k + k - 1}{\prod_{i = 1, i \neq k}^{d} \mu_i - i - \mu_k + k} = \frac{\prod_{i=1}^{d-1} \nu_i - i -\lambda_k + k}{\prod_{i = 1, i \neq k}^{d} \mu_i - i - \lambda_k + k + 1},
\end{equation}
where $\mu + e_k = \lambda$ and $\nu \sqsubseteq \mu$.

In the following, we say that a partition $\mu$ is ``parent'' of a partition $\lambda$ (or $\lambda$ is ``child'' of $\mu$) if $\mu$ can be obtained from $\lambda$ by removing some box from $\lambda$. 

First note, that there are only two possible parents of $\lambda_1$:
\begin{equation}
\mu_1=(n-k,1^{k-1},0^{d-k}),\qquad
\mu_2=(n-k-1,1^{k},0^{d-k-1}).
\end{equation}
Each such parent $\mu$ has $\mu+e_k = \lambda_1$ as one child and other child is $\mu+e_{k'}$ with $c_{\mu+e_{k'}}=0$.
Hence
\begin{equation}
E_{\mu}(\nu)=c_{\lambda_1}\,w_{k}^{(\mu,\nu)}\ge 0.
\end{equation}
Thus no parent of $\lambda_1$ can lower the minimum below zero. Negative values come only from parents of $\lambda_2$. Now let's consider two different cases.

\begin{itemize}
    \item \textit{Generic case $n \ge k+3$.} The partition $\mu_3=(n-k-2,1^{k+1},0^{d-k-2})$ is a valid parent of $\lambda_2$, since adding a box to row $1$ yields
$\mu_3+e_{1}=\lambda_2$.
All other children of $\mu_3$ have $c_\lambda=0$. Choose the interlacing partition
$\nu=(n-k-2,1^{k},0^{d-k-2})$,
which saturates all inequalities $\mu_{3,i} \ge \nu_i \ge \mu_{3,i+1}$.
Substituting this $\nu$ into the product formula, one finds
$w_{1}^{(\mu_3,\nu)}=1$.
Thus $\mu_3$ can place all weight on the negative child $\lambda_2$, giving
\begin{equation}
E_{\mu_3}^{\min}
=
-\frac{1}{d_{\lambda_2}^2}
=
-\frac{1}{\binom{n-1}{k+1}^2}.
\end{equation}
Since every parent of $\lambda_1$ contributes $\ge 0$, and no other parent of $\lambda_2$
can achieve a value \(< -1/d_{\lambda_2}^2\), we obtain
\begin{equation}
\alpha_{n-1|1}(W^{(k)}) = -\frac{1}{\binom{n-1}{k+1}^{2}}
\qquad\text{for } n\ge k+3.
\end{equation}

\item \textit{Boundary case $n=k+2$.}
Here 
\begin{equation}
\lambda_1=(2,1^{k},0^{d-k-1}),\qquad
\lambda_2=(1^{k+2},0^{d-k-2}),
\end{equation}
and the parent $\mu_3$ does not exist.
The only parent for $\lambda_2$ is also a parent for $\lambda_1$:
\begin{equation}
\mu_2=(1^{k+1},0^{d-k-1}).
\end{equation}
We choose $\nu=(1^{k},0,\dots,0)$. A direct evaluation of the coefficients yields
\begin{equation}
w_{1}^{(\mu_2,\nu)}=\frac{k+1}{k+2},
\qquad
w_{k+2}^{(\mu_2,\nu)}=\frac{1}{k+2}.
\end{equation}
Using $d_{\lambda_1}=\binom{k+1}{k}=k+1$ and $d_{\lambda_2}=1$, we obtain
\begin{equation}
E_{\mu_2}^{\min}
=
\frac{1}{(k+1)^2}\cdot\frac{k+1}{k+2}
\;-\;
1\cdot \frac{1}{k+2}
=
-\frac{k}{(k+1)(k+2)}.
\end{equation}
Since every parent of $\lambda_1$ gives $\ge 0$, this is the global minimum:
\begin{equation}
\alpha_{n-1|1}(W^{(k)})
=
-\frac{k}{(k+1)(k+2)}
\qquad\text{for }n=k+2.
\end{equation}
\end{itemize}
Combining the cases, we get
\begin{equation}
\alpha_{n-1|1}(W^{(k)})
=
\begin{cases}
-\dfrac{1}{\binom{n-1}{k+1}^{\,2}}, & n\ge k+3,\\[6pt]
-\dfrac{k}{(k+1)(k+2)}, & n=k+2,
\end{cases}
\end{equation}
as claimed.
\end{proof}

The case $n=k+2$ provides valid witnesses because the Hook inequalities hold between immanant inequalities~\cite{pate1992immanantHook}, and the resulting witness is non-PSD. So it allows to detect high-dimensional $[n-1|1]$-separable states with $d \geq n$.

\begin{cor}
For any $a \geq b$, $n = a + b$, $a-b \geq 2k > 0$ consider a qubit operator
\begin{equation}
W =
\frac{\Pi_{(a,b)}}{d_{(a,b)}^{\,2}} - \frac{\Pi_{(a-k,b+k)}}{d_{(a-k,b+k)}^{\,2}}, \qquad 
d_{(n-r,r)} = \binom{n}{r}-\binom{n}{r-1} =
\binom{n}{r}\frac{n-2r+1}{n-r+1},
\end{equation}
Then,
\begin{equation}
\alpha_{n-1|1}(W) =
\begin{cases}
-\dfrac{1}{d_{(a-k,b+k)}^{\,2}}, 
& \text{if } a-b \geq 2k+1, \\[10pt]
-\dfrac{1}{2\,d_{(a-k,b+k)}^{\,2}}, 
& \text{if } a-b=2k, \, k > 1 \\[10pt]
\dfrac{1}{2}\of*{
\dfrac{1}{d_{(a,b)}^{\,2}}
-\dfrac{1}{d_{(a-1,b+1)}^{\,2}}
},
& \text{if } a-b=2, \, k = 1.
\end{cases}
\end{equation}
\end{cor}
\begin{proof}
Let $d=2$ and $n=a+b$ with $a\ge b\ge 0$. Consider the two--row partitions
$\lambda_1=(a,b)$ and $\lambda_2=(a-k,b+k)$, where $k\in\mathbb{Z}_{\ge 1}$
and $\Delta\coloneqq a-b\ge 2k$ so that $\lambda_2$ is a valid two--row shape.
We use the $[n-1|1]$ minimization formula from the Theorem~\ref{thm:exact_(1,n-1)_sep}:
\begin{equation}\label{eq:opt_func_cor_2}
\alpha_{n-1|1}(W)
=
\min_{\mu\vdash_2\, (n-1)}
\ \min_{\nu\sqsubseteq \mu}
\ \sum_{e_k \in\mathrm{AC}_d(\mu)} c_{\mu+e_k} \, w_k^{(\mu,\nu)}, \qquad w_{k}^{(\mu,\nu)} = \frac{\prod_{i=1}^{d-1} \nu_i - i - \mu_k + k - 1}{\prod_{i = 1, i \neq k}^{d} \mu_i - i - \mu_k + k},
\end{equation}
where $c_{\lambda_1}=d_{\lambda_1}^{-2}$, $c_{\lambda_2}=-d_{\lambda_2}^{-2}$ and all other $c_\lambda=0$.
For qubits ($d=2$), every parent $\mu=(x,y)$ with $x\ge y\ge 0$ has exactly two
addable corners: $e_1$ adds to row $1$, $e_2$ adds to row $2$, producing the children
$\mu+e_1=(x+1,y)$ and $\mu+e_2=(x,y+1)$. The interlacing diagram $\nu$ is parametrized by a single integer $t$ with $y \le t\le x$, so the coefficients $w_k := w^{((x,y),(t))}_k$ specialize to
\begin{equation}
w_1=\frac{x-t+1}{x-y+1},\qquad
w_2=\frac{t-y}{x-y+1},\qquad
w_1+w_2=1.
\end{equation}
Only parents that can reach at least one of $\lambda_1,\lambda_2$ by adding one box can contribute to the sum in Eq.~(\ref{eq:opt_func_cor_2}). In the following, we will analyse the function
\begin{equation}
    E_{\mu^{(k)}}(t) := c_{\mu+e_k} \, w_k^{(\mu,(t))}.
\end{equation}

\emph{(i) Parents of $\lambda_1=(a,b)$.}
Removing from row $1$ gives $\mu_{1}^{(1)}=(a-1,b)$ (children $\lambda_1$ and $(a-1,b+1)$);
removing from row $2$ (if $b\ge 1$) gives $\mu_{1}^{(2)}=(a,b-1)$ (children $(a+1,b-1)$ and $\lambda_1$).
In both cases, the other child has coefficient $0$ in $W$, so
\begin{equation}
E_{\mu_{1}^{(k)}}(t) = \frac{1}{d_{\lambda_1}^{2}}\,w_{k} \ \ge\ 0.
\end{equation}
Hence parents of $\lambda_1$ cannot lower the minimum below $0$.

\emph{(ii) Parents of $\lambda_2=(a-k,b+k)$.}
Removing from row $1$ gives $\mu_{2}^{(1)}=(a-k-1,b+k)$ (children $\lambda_2$ and $(a-k-1,b+k+1)$);
removing from row $2$ gives $\mu_{2}^{(2)}=(a-k,b+k-1)$ (children $(a-k+1,b+k-1)$ and $\lambda_2$).
In both cases the other child has coefficient $0$. Therefore
\begin{equation}
E_{\mu_{2}^{(k)}}(t)= -\frac{1}{d_{\lambda_2}^{2}}\,w_k \leq 0,
\end{equation}
and to minimize we aim to maximize the coefficient $w_k$.
To do that, we write $\Delta = a-b$, and consider several cases:
\begin{itemize}
\item \textit{Case A: $\Delta \ge 2k+1$.} Then $\mu_{2}^{(1)}=(a-k-1,b+k)$ is a valid parent (since $a-k-1\ge b+k$). Choosing $t=b+k$ gives $w_1=1$ and $w_2=0$, hence
\begin{equation}
E_{\mu_{2}^{(1)}}^{\min}=-\frac{1}{d_{\lambda_2}^{2}}.
\end{equation}
Since parents of $\lambda_1$ contribute $\ge 0$, we obtain
\begin{equation}
\alpha_{n-1|1}(W)=-\frac{1}{d_{(a-k,b+k)}^{2}}
\qquad\text{for }\Delta\ge 2k+1.
\end{equation}

\item \textit{Case B: $\Delta=2k> 2$.} Here $\mu_{2}^{(1)}$ is invalid (it would have first part $a-k-1<b+k$), but $\mu_{2}^{(2)}=(a-k,b+k-1)$ is valid with $x-y=(a-k)-(b+k-1)=1$.
Thus the maximal coefficient of the row--2 child $\lambda_2$ is achieved at $t=x=a-k$ and equals
\begin{equation}
w_2^{\max}=\frac{x-y}{x-y+1}=\frac{1}{2}.
\end{equation}
Hence
\begin{equation}
E_{\mu_{2}^{(2)}}^{\min}=-\frac{1}{2}\cdot \frac{1}{d_{\lambda_2}^{2}},
\end{equation}
and therefore
\begin{equation}
\alpha_{n-1|1}(W)=-\frac{1}{2\,d_{(a-k,b+k)}^{2}}
\qquad\text{for }\Delta=2k > 2.
\end{equation}

\item \textit{Case C: $k=1$ and $\Delta=2$.} Here $\lambda_2=(a-1,b+1)$ and the parent $\mu=(a-1,b)$ has both children. Since $x-y=(a-1)-b=1$, the interlacing choice $t=x$ yields $w_1=w_2=\tfrac12$, so
\begin{equation}
\alpha_{n-1|1}(W)=
\frac{1}{2}\left(\frac{1}{d_{(a,b)}^{2}}-\frac{1}{d_{(a-1,b+1)}^{2}}\right)
\qquad\text{for }\Delta=2.
\end{equation}
\end{itemize}

Collecting the three cases gives
\begin{equation}
\alpha_{n-1|1}(W)
=
\begin{cases}
-\dfrac{1}{d_{(a-k,b+k)}^{2}}, & \Delta\ge 2k+1, \\[6pt]
-\dfrac{1}{2\,d_{(a-k,b+k)}^{2}}, & \Delta=2k, \\[6pt]
\dfrac{1}{2}\!\left(\dfrac{1}{d_{(a,b)}^{2}}-\dfrac{1}{d_{(a-1,b+1)}^{2}}\right),
 & k=1,\ \Delta=2~.
\end{cases}
\end{equation}
Finally, for the adjacent pair $(a,b)=(n,0)$ and $(a-1,b+1)=(n-1,1)$ one has
$d_{(n,0)}=1$ and $d_{(n-1,1)}=n-1$, and Case A gives
\begin{equation}
\alpha_{n-1|1}(W)=-\frac{1}{(n-1)^2}\qquad(n\ge 3),
\end{equation}
while $\alpha_{1|1}=0$ for $n=2$.
\end{proof}

\section{Equivalence to immanant inequalities}

Two well-known families of immanant inequalities for $n\times n$ matrices are:
\begin{itemize}
    \item The {\em hook rule}~\cite{pate1992immanantHook}: given two partitions of the form $\lambda=(\lambda_1,1^{n-\lambda_1})$ and $\mu=(\mu_1,1^{n-\mu_1})$ with $\lambda_1\geq\mu_1$, then
\begin{equation}
    \frac{\imm_\lambda(M)}{\chi_\lambda(\id)} \geq \frac{\imm_\mu(M)}{\chi_\mu(\id)}\,.
\end{equation}
\item The {\em Lieb's permanent dominance} (conjectured) and {\em Schur's inequality}: it is long believed that the permanent is the largest normalized immanant~\cite{LiebPermanent1966}; and the determinant was shown to be smaller than any normalized immanant~\cite{schur1918DetSmall},
\begin{equation}
    \per(M)\geq \frac{\imm_\lambda(M)}{\chi_\lambda(\id)}\geq \det(M)
\end{equation}
for all $\lambda\vdash n$. The permanent dominance has been proven true for matrices of size $n\leq 13$~\cite{pate1999tensorPerDomN13}, and therefore can be used to construct entanglement witnesses for qudit dimensions up to $13$.
\end{itemize}

We follow~\cite{MaassenSlides}. Any positive semidefinite, rank-$r$ matrix $G$ can be written as $G=VV^\dag$, where $V$ has columns $\ket{v_i}\in\C^r$. Therefore, we have
\begin{align}
    \imm_\lambda(G)&=\sum_{\pi\in S_n}\chi_\lambda(\pi)\prod_{i=1}^n G_{i,\pi(i)}\\
    &=\sum_{\pi\in S_n}\chi_\lambda(\pi)\prod_{i=1}^n \bra{v_i}v_{\pi(i)}\rangle\\
    &=\sum_{\pi\in S_n}\chi_\lambda(\pi)\tr\Big (\pi_d \bigotimes_{i=1}^n\ket{v_{i}}\bra{v_i}\Big )\\
    &=\tr\bigg (\frac{n!}{\chi_\lambda(\id)} \Pi_\lambda \bigotimes_{i=1}^n\ket{v_{i}}\bra{v_i}\bigg )\,.
\end{align}
It follows that
\begin{align}
    \sum_{\lambda\vdash n}a_\lambda\imm(G) &= \sum_{\lambda\vdash n}\frac{n! a_\lambda}{\chi_\lambda(\id)} \tr\Big (\Pi_\lambda \bigotimes_{i=1}^n\ket{v_{i}}\bra{v_i}\Big )\\
    &= \sum_{\lambda\vdash n}\frac{n! a_\lambda}{\chi_\lambda(\id)} \tr\Big (\Pi_\lambda \varrho\Big )
\end{align}
where $\varrho$ is a product state. It is clear that if the left hand side is nonnegative on positive semidefinite matrices $G$, then the right hand side is nonntegative on product states, and by convexity also on separable states. The converse direction holds analogously, since all steps are equalities.

For instance, the immanant inequalities and corresponding Young projector orderings above apply from $4\times 4$ as:
\begin{align}
    &\per(M)\geq \frac{\imm_{(3,1)}(M)}{3} \geq \frac{\imm_{(2,1^2)}(M)}{3}\geq\det(M) \\
    &\langle \Pi_{(4)}\rangle_\varrho \geq \frac{\langle \Pi_{(3,1)}\rangle_\varrho }{9} \geq \frac{\langle \Pi_{(2,1^2)}\rangle_\varrho}{9}\geq \langle \Pi_{(1^4)}\rangle_\varrho\quad\forall\,\varrho\in\text{SEP}
\end{align}
And:
\begin{align}
    &\per(M)\geq \frac{\imm_{(2,2)}(M)}{2} \geq\det(M) \\
    &\langle \Pi_{(4)}\rangle_\varrho  \geq \frac{\langle \Pi_{(2,2)}\rangle_\varrho}{4}\geq \langle \Pi_{(1^4)}\rangle_\varrho\quad\forall\,\varrho\in\text{SEP}\,.
\end{align}
To our knowledge, this is a complete list of all known immanant inequalities before this work. For completeness we provide the table of characters for 4-partite systems below.
\begin{table}[h!]
    \centering
    \begin{tabular}{c | c c c c c}
     & $(\id)$ & $(12)(3)(4)$ & $(12)(34)$ & $(123)(4)$ & $(1234)$ \\
     \hline
        $\lambda=(4)$ & $1$ & $1$ & $1$ & $1$ & $1$  \\
        $\lambda=(3,1)$ & $3$ & $1$ & $-1$ & $0$ & $-1$  \\
        $\lambda=(2,2)$ & $2$ & $0$ & $2$ & $-1$ & $0$  \\
        $\lambda=(2,1^2)$ & $3$ & $-1$ & $-1$ & $0$ & $1$  \\
        $\lambda=(1^4)$ & $1$ & $-1$ & $1$ & $1$ & $-1$  \\
    \end{tabular}
    \caption{{\bf Characters of each permutation type $(\pi)$ for each irrep partition $\lambda$.}}
    \label{tab:my_label}
\end{table}

\end{widetext}
\end{document}